\documentclass[journal,transmag]{IEEEtran}

\ifCLASSINFOpdf
\else
\fi

\usepackage[cmex10]{amsmath}

\usepackage{color}
\usepackage{soul}

\usepackage{tikz}
\usetikzlibrary{decorations.pathmorphing} 
\usetikzlibrary{fit}					
\usetikzlibrary{backgrounds}	

\usepackage{longtable}
\usepackage{booktabs}
\usepackage{forloop}
\usepackage{xstring}
\usepackage{caption}
 
\usepackage{multirow}
\setulcolor{red}
\sethlcolor{yellow}

\makeatletter
\def\ps@headings{%
\def\@oddhead{\mbox{}\scriptsize\rightmark \hfil \thepage}%
\def\@evenhead{\scriptsize\thepage \hfil \leftmark\mbox{}}%
\def\@oddfoot{}%
\def\@evenfoot{}}
\makeatother
\def\squareforqed{\hbox{\rlap{$\sqcap$}$\sqcup$}}
\def\qed{\ifmmode\squarefored\else{\unskip\nobreak\hfil
\penalty50\hskip1em\null\nobreak\hfil\squareforqed
\parfillskip=0pt\finalhyphendemerits=0\endgraf}\fi}

\usepackage{amsthm}
\usepackage{graphicx, amssymb, subfigure,setspace}
\usepackage{times}

\usepackage{wrapfig}
\usepackage{textcomp}
\usepackage{array}
\usepackage{algorithmic}
\usepackage[vlined, ruled, linesnumbered]{algorithm2e}

\usepackage[long]{optional}

\usepackage{comment}

\newtheorem{lemma}{Lemma}
\newtheorem{theorem}{Theorem}
\newtheorem{corollary}{Corollary}
\newtheorem{proposition}{Proposition}

\newenvironment{definition}[1][Definition]{\begin{trivlist}
\item[\hskip \labelsep {\bfseries #1}]}{\end{trivlist}}

\usepackage{arydshln}

\newcommand{\xx}[1]{\textcolor{blue}{{\bf } #1}}

\begin{document}

\title{A Truthful $(1-\epsilon)$-Optimal Mechanism for On-demand Cloud Resource Provisioning}

\author{\IEEEauthorblockN{Xiaoxi Zhang\IEEEauthorrefmark{1},~\IEEEmembership{Student Member,~IEEE,~ACM},
Chuan Wu\IEEEauthorrefmark{1},~\IEEEmembership{Member,~IEEE,~ACM},\\
Zongpeng Li\IEEEauthorrefmark{2},~\IEEEmembership{Member,~IEEE,~ACM}, and
Francis C.M. Lau\IEEEauthorrefmark{1},~\IEEEmembership{Senior Member,~IEEE}}
\IEEEauthorblockA{\IEEEauthorrefmark{1}Department of Computer Science,
The University of Hong Kong, Hong Kong}
\IEEEauthorblockA{\IEEEauthorrefmark{2}Department of Computer Science, University of Calgary, Calgary, Canada}
}


\IEEEtitleabstractindextext{
\begin{abstract}
On-demand resource provisioning in cloud computing provides tailor-made resource packages (typically in the form of VMs) to meet users' demands. Public clouds nowadays provide more and more elaborated types of VMs, but have yet to offer the most flexible dynamic VM assembly, which is partly due to the lack of a mature mechanism for pricing tailor-made VMs on the spot. This work proposes an efficient randomized auction mechanism based on a novel application of smoothed analysis and randomized reduction, for dynamic VM provisioning and pricing in geo-distributed cloud data centers. This auction, to the best of our knowledge, is the first one in literature that achieves (i) truthfulness in expectation, (ii) polynomial running time in expectation, and (iii) $(1-\epsilon)$-optimal social welfare in expectation for resource allocation, where $\epsilon$ can be arbitrarily close to $0$.
Our mechanism consists of three modules: (1) an exact algorithm to solve the NP-hard social welfare maximization problem, which runs in polynomial time in expectation,
(2) a perturbation-based randomized resource allocation scheme which produces a VM provisioning solution that is $(1-\epsilon)$-optimal
and (3) an auction mechanism that applies the perturbation-based scheme for dynamic VM provisioning and prices the customized VMs using a randomized VCG payment, with a guarantee in truthfulness in expectation.

We validate the efficacy of the mechanism through careful theoretical analysis and trace-driven simulations.

\end{abstract}

}

\maketitle

\IEEEdisplaynontitleabstractindextext

\IEEEpeerreviewmaketitle

\section{Introduction}
\label{sec:intro}

Cloud computing services have been proliferating in today's Internet for the past decade. They create a shift in resource provisioning from on-premise hardware to shared resource pools accessible over the Internet. 
To be flexible at meeting users' resource demands, 
leading cloud platforms such as Amazon EC2 \cite{amazon:ec2}, Microsoft Azure \cite{microsoft:azure} and GoGrid \cite{GoGrid} exploit advanced virtualization technologies to pack resources (CPU, RAM, and Disk Storage) into virtual machine (VM) instances of various types. Undoubtedly, the more variety of VM types they can provide, the better they could meet the wide range of users' demands. For example, Amazon EC2 has been expanding the variety of VM instances they provide, which now spans $9$ categories and $39$ types \cite{amazon:ec2instances}. However, the increased variety on the provider's side still often falls short of addressing user needs precisely, which could lead to a waste of resources and an unjustifiably inflated payment by the users. For example, suppose a user  needs to run a computationally intensive job ({\em e.g.}, a MapReduce job) by acquiring $16$ vCPU units and $16$ GB memory \cite{mapreduce:memoryconfig} in EC2's Singapore data center, to process $160$ GB usage data. 
The best offer Amazon EC2 can make is a {\tt c3.4xlarge} instance, which unfortunately is far from a perfect match, leading to a waste of roughly half of the allocated memory and SSD storage.

Current virtualization technology is in fact ready for real-time, on-demand VM partitioning and provisioning ({\em e.g.}, by utilizing credit-based CPU scheduler and memory ballooning \cite{Xen-SOSP03}). What is not ready, however, 
is an effective pricing mechanism to charge for those customized VMs on the spot. The current representative pricing models, {\em e.g.}, long-term reservation, fixed on-demand instance pricing and spot instance pricing employed by Amazon EC2, are not suitable for dynamically assembled VMs. Under fixed pricing, it is impossible for the cloud provider to come up with
the appropriate prices, {\em a priori}, for any VM type that could possibly be assembled according to the user's needs. Furthermore, fixed pricing fails to cater to the ever-changing supply and demand in the market; either overpricing or underpricing would jeopardize the social welfare of the overall system as well as the provider's revenue. Amazon's spot instances market \cite{amazon:spot} represents the first attempt at a more market-driven
pricing system, which, however, 
comes without any guarantee of truthfulness or SLA
 \cite{Wang:2012kt}\cite{Wang:2013dy}. Some recent work further studied auction mechanism design for cloud resource provisioning from different perspectives \cite{Wang:2013dy}\cite{Zaman:2012:CAM:2310096.2310159}\cite{Zaman:2013:CAA:2442165.2442286}. 
However, most of them model VMs as type-oblivious commodities,
and therefore fail to provision dynamically assembled VMs properly.

Aside from pricing, the challenge of packing available resources to maximally cater to users' VM demands translates into an NP-hard combinatorial optimization problem, which presents a tough challenge in VM auction design. The VCG mechanism \cite{vickrey:vcg}, essentially the only type of auction that guarantees both truthfulness and economic efficiency (social welfare maximization), requires an
exact optimal allocation. When
polynomial-time approximation algorithms are applied instead, VCG loses its truthfulness property \cite{Mu'alem2008612}. To achieve truthfulness with an approximation algorithm, researchers have exploited the concept of critical bids \cite{lehmann:criticalvalue}, or resorted to some LP decomposition techniques \cite{lavi:lpdecomposition}\cite{Zhang:2014jo}\cite{Shi:sigmetrics14}. The approximation ratios of these auctions with respect to social welfare optimality depend on the efficiency of the approximation algorithm employed, which is typically much larger than $1$ \cite{Zhang:2014jo}\cite{Shi:sigmetrics14}.

This work aims to leverage the state-of-the-art techniques in smoothed analysis \cite{Spielman:2001:SAA:380752.380813}\cite{Brunsch:improved_smoothed} and randomized reduction, to design a highly efficient randomized auction mechanism for the provisioning and pricing of customized VM instances in a geo-distributed cloud. The resulting combinatorial VM auction is sufficiently expressive for cloud users to request the necessary custom-made VM instances in bundles in different data centers for their job execution. To the best of our knowledge, this is the first VM auction that achieves (i) truthfulness (in expectation), (ii) polynomial running time (in expectation), and (iii) $(1-\epsilon)$-optimal social welfare (in expectation) for resource allocation in a geo-distributed cloud, where $\epsilon\in (0,1)$ is a tunable parameter that can approach zero.

Our proposed auction mechanism consists of three main modules: (1) an exact algorithm to solve the NP-hard social welfare maximization problem, which runs in polynomial time in expectation based on smoothed analysis. It serves as the basis for resource allocation in (2); (2) a perturbation-based randomized resource allocation scheme that produces a VM provisioning solution achieving $(1-\epsilon)$-optimal social welfare in expectation; and (3) an auction mechanism that applies the
perturbation-based scheme to dynamic VM provisioning, and prices the customized VMs using a randomized VCG payment, which guarantees truthfulness in expectation. Detailed steps are as followed.


{\em First}, we formulate the social welfare optimization problem as an integer linear program and then prove its NP-hardness. Based on smoothed analysis, we randomly perturb the objective function and the packing constraints following a well designed perturbation framework, and propose an exact dynamic programming based algorithm to solve the perturbed problem. The algorithm finds a feasible solution to the original, unperturbed problem within polynomial running time in expectation. Furthermore, a transformation of this feasible solution yields a fractional solution to the original problem, which achieves $(1-\epsilon)$-optimal social welfare in expectation.

{\em Next}, we design a randomized resource allocation scheme, which outputs the allocation solution of the auction following a well designed distribution over a set of feasible solutions of the social welfare maximization problem, including the feasible solution produced in the above step. By designing the distribution in close connection with the perturbation framework, we are able to show that the expectation of such a randomized solution equals the fractional solution mentioned above, and hence it achieves $(1-\epsilon)$-optimal social welfare in expectation.

{\em Finally}, we combine the randomized resource allocation scheme with a randomized VCG payment, and complete our auction design with truthfulness guaranteed in expectation.


Our mechanism design yields some interesting results: (i) For the social welfare maximization problem we formulate, {\bf even if truthful bids are given for free}, no (deterministic) polynomial-time algorithm can guarantee $(1-\epsilon)$-approximation for arbitrarily tunable $\epsilon$ \cite{Gens:1980:CAA:1008861.1008867}. (ii) The (randomized) VM auction designed in our work is both polynomial-time and $(1-\epsilon)$-optimal in expectation, and can simultaneously elicit truthful bids from selfish cloud users.


The strong properties above guaranteed by our VM auction are made possible by unleashing the power of randomization, through the art of calculated random perturbation (for computational efficiency) and associated perturbation (for truthfulness). While there exists separate literature on applying randomization for efficient algorithm design and for truthful mechanism design respectively, to the best of our knowledge, this work is the first of its kind that applies the same carefully prepared randomization scheme twice in subtly different forms and in a coordinated fashion to achieve polynomial algorithm complexity and truthful mechanism design in the same auction framework. We believe that this new technique can be generalized to be applicable to a rich class of combinatorial auctions in which social welfare maximization can be modeled as a linear integer program that is NP-hard (otherwise our technique is unnecessary) but not too hard (which admits a smoothed polynomial time algorithm) to solve.


We discuss related work in Sec.~\ref{sec:related} and present the system model in Sec.~\ref{sec:model}. Sec.~\ref{sec:auction} gives the complete auction design. Sec.~\ref{sec:simulation} presents trace-driven simulation studies and Sec.~\ref{sec:conclusion} concludes the paper.

\section{Related Work}
\label{sec:related}

Resource provisioning in cloud computing has been extensively studied with different focuses. 
Beloglazov {\em et al.}~\cite{Beloglazov2012755} aim at minimizing the energy consumption in computing task scheduling. Alicherry {\em et al.}~\cite{Alicherry:2012ej} study VM allocation in distributed cloud systems, taking into consideration the communication cost. Joe-Wong {\em et al.}~\cite{JoeWong:2012bc} seek to balance efficiency and fairness for allocating resources of multiple types. None of them however focus on dynamic VM assembly and provisioning, which is the focus of our work.

Auction mechanisms have been applied to achieve efficient resource allocation in cloud systems. 
Zaman {\em et al.}~\cite{Zaman:2012:CAM:2310096.2310159} design a truthful auction based on an approximation algorithm for resource allocation, but without proving the performance of the resource allocation algorithm.
Zaman {\em et al.}~\cite{Zaman:2013:CAA:2442165.2442286} also presents an auction-based VM allocation but focuses on static resource provisioning and only guarantees a large approximation ratio. 
Wang {\em et al.}~\cite{Wang:2012kt} propose a truthful VM auction based on a greedy allocation algorithm and a well-designed payment method; the derived allocation solution approximates the optimal solution with an approximation ratio which depends on the number of VMs. Zhang {\em et al.}~\cite{Zhang:2013bi} and Wang {\em et al.}~\cite{Wang:2013dy} design online cloud auctions but they only consider a single type of VM instances, ignoring dynamic provisioning of different VMs. Similar to our work, Zhang {\em et al.}~\cite{Zhang:2014jo} and Shi {\em et al.}~\cite{Shi:sigmetrics14} address dynamic VM provisioning, and design truthful auctions by applying an LP decomposition technique, which achieve 2.72- and 3.30-approximation of optimal social welfare, respectively.   
Mashayekhy et al.~\cite{Mashayekhy:ptas} also consider heterogeneous resources in a cloud auction. Their proposed allocation rule is proved to be a PTAS, which finds a partial allocation solution first and then allocates through dynamic programming the remaining resources to unprovisioned users. Unfortunately, the running time of their PTAS mechanism can be exponential in $1/\epsilon$, where $\epsilon$ denotes the approximation error.
Shi {\em et al.}~\cite{wshi-iwqos14} design online VM auctions via a pricing-curve-based method. However, the resources which have been allocated and paid by the users could be cancelled by the cloud provider in their problem model, which is not practical and user-friendly. One of the latest VM auctions proposed by Zhang {\em et al.}~\cite{xzhang:sigmetrics15} investigates social welfare and profit maximization in online VM auctions, where the future demand of each type of resource could be reserved in a customized amount and duration. In their work, the performance of the mechanisms relies on an assumption that each user's demand for each resource is extremely small compared to the corresponding capacity. Unlike theirs, we study the auction in an offline version where the small bid assumption is much milder than theirs.  
Additionally, this work departs from the existing literature by applying smoothed analysis and randomized reduction techniques to design a randomized auction, which achieves much better approximation to the optimal solution, {\em i.e.}, $(1-\epsilon)$-optimal social welfare (where $\epsilon$ can be very close to zero), while retaining truthfulness and computation efficiency in expectation.


A key technique we adopt in this paper is a novel use of smoothed analysis in designing an algorithm to produce a good solution to the social welfare maximization in polynomial time in expectation. Smoothed analysis is a technique for analyzing the time complexity of an algorithm for an NP-hard problem, that exactly solves a perturbed instance of the problem based on a small, random perturbation, in order to show 
that the algorithm can be efficient in expectation despite its possible worst-case complexity \cite{Spielman:2001:SAA:380752.380813}. It has been argued that complexity analysis on the expectation over some distribution of the instances is more convincing than that of the average case, and more practical than that of the worst case \cite{Spielman:2001:SAA:380752.380813}. Smoothed analysis has been applied recently in areas such as combinatorial programming \cite{Roglin:2009ku}, computational geometry \cite{Arthur:2006hh}, game analysis \cite{Chen:2006ii}.
Dughmi {\em et al.}~\cite{conf:focs:DughmiR10} focus on social welfare maximization problems with an FPTAS, and design a randomized reduction method to convert the FPTAS into a truthful mechanism. Unlike theirs, our NP-hard social welfare maximization problem does not have a deterministic FPTAS; even so, we are able to show, surprisingly, that we can still achieve a randomized, truthful, $(1-\epsilon)$-approximation mechanism with expected polynomial complexity by applying the carefully designed permutation framework.

\section{System Model}
\label{sec:model}

\subsection{Cloud Resource Aucion}
We consider an IaaS cloud system providing a pool of geo-distributed resources ({\em e.g.,} CPU, RAM and Disk storage) through servers deployed in multiple data centers. The different types of resources are packed into heterogeneous VMs for lease to cloud users. Suppose there are $D$ data centers in the cloud and $M$ types of VMs in total can be assembled to offer (which can be potentially a very large number)\footnote{Note that we allow flexible VM assembly on demand and the numbering of VM types is purely for the ease of presentation.} composed of $K$ types of resources. Let $[X]$ denote the set of integers $\{1,2,\ldots,X\}$, $d$ index the data centers in the cloud system and $m$ index the types of VM.  Each VM of type $m\in [M]$ consumes an $r^k_m$ amount of type-$k$ resource, for all $k\in[K]$. Each data center $d\in [D]$ has a capacity $c_{dk}$ for each resource of type $k\in[K]$.

The cloud provider acts as an auctioneer and sells custom-made VMs to $W$ cloud users through auctions. 
The cloud users request resources in the form of VMs through bidding in the auction. Let $N$ be the total number of bids submitted by all the cloud users. We use $i\in[N]$ to index the bids. 
Each cloud user $w \in  [W]$ is allowed to submit multiple bids, but at most one bid can be successful. 
This assumption is reasonable given that any need for concurrently acquiring VM bundles in two or more bids can be expressed as a separate bid with a combined bundle. 
Let $\mathcal{B}_w$ denote the set of bids submitted by user $w$. 
Each bid $B_i$ contains a list of requested VMs of different types and location preferences, {\em i.e.}, $q_{md}^i$ VMs of type-$m$ in data center $d$, $\forall m\in[M], d\in[D]$, and a bidding price $b_i$, {\em i.e.}, the reported valuation of the resource combination required in $B_i$.
More specifically, each bid $i\in[N]$ can be formulated as follows:
\begin{equation*}
B_i=\{b_i, \{q_{md}^{i}\}_{m\in[M], d\in [D]}\}.
\end{equation*}
\vspace{-2mm}

%
\noindent For ease of problem formulation, we use $R_{i}^{kd}=\sum_{m=1}^{M}q_{md}^{i}r_m^k$ to denote the overal amount of type-$k$ resource required by bid $i$ in data center $d$.


Upon receiving user bids, the cloud provider computes the outcome of the auction, including (i) the resource allocation scheme, $\vec{x}=\{x_1,...,x_N\}$, where binary variable $x_i$ is $1$ if bid $i$ is successful and $0$ otherwise, and (ii) a payment $p_i$ for each winning bid $i$.
Let $v_i$ denote the true valuation of the bidder submitting bid $i$. The utility $u_i$ acquired due to this bid is then:
\begin{equation*}
u_i(B_i,\mathcal{B}_{-i})= \left\{\begin{array}{l l}
	v_i-p_i & \quad \text{if $B_i$ is accepted}\\
    0 & \quad \text{otherwise}
  \end{array}\right.
\end{equation*}

\noindent where $\mathcal{B}_{-i}$ is the set of all bids in the auction except $B_i$.
An illustration of the system is given in Fig.~\ref{illustration:auction}.


We summarize important notations in Table \ref{tab:notation} for ease of reference.


\begin{figure}[!t]
\begin{center}
  \includegraphics[width=0.48\textwidth]{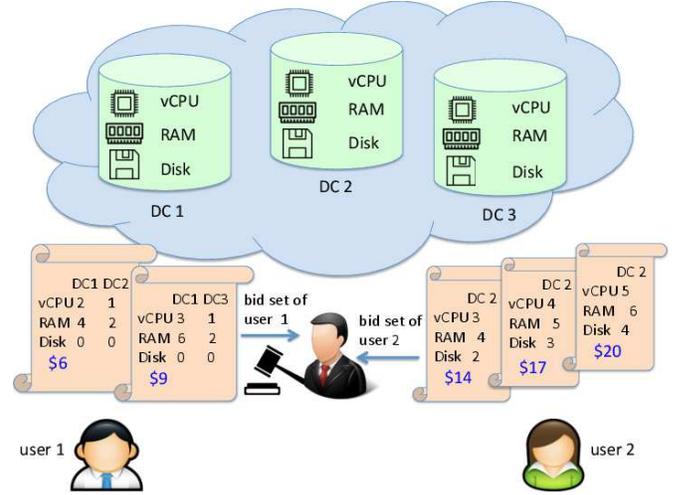}
  \caption{An illustration of the VM auction.}
  \label{illustration:auction}
\end{center}
\end{figure}

\begin{table}[t!]
\caption{Notation}
\label{tab:notation}
\centering
\begin{tabular}{|r|l|l|l|}
\hline
$W$  & $\#$ of users &$N$ & $\#$ of bids \\
\hline
$M$  & $\#$ of VM types &$K$ & $\#$ of resource types \\
\hline
$D$  & $\#$ of data centers& $B_i$ & the $i$th bid \\
\hline
$v_i$ & true valuation of bid $i$ & $u_i$ & utility of bid $i$\\
\hline
$P$ & perturbation matrix & $\mathcal{B}_w$ & bid set of user $w$\\
\hline
${b}_i$ & bidding price of bid $i$ & $\hat{b}_i$ & perturbed ${b}_i$\\
\end{tabular}
\begin{tabular}{|r|l|}
\hline


$\epsilon$ & parameter in $(0,1)$\\
\hline

$r_m^k$   & amount of type-$k$ resource in a type-$m$ VM \\
\hline

$q^i_{md}$ & $\#$ of type-$m$ VMs in DC $d$ requested in bid $i$ \\
\hline

$R_{i}^{kd}$ & demand of type-$k$ resource in DC $d$ in bid $i$\\
\hline

$\hat{R}_{i}^{kd}$ & perturbed $R^{kd}_i$\\
\hline



$c_{kd}$  & capacity of type-$k$ resource in DC $d$ \\
\hline

$s(\vec{x})$ & social welfare under allocation solution $\vec{x}$ \\
\hline


$C_{kd}(\vec{x})$ & demand for type-$k$ resource in DC $d$ under $\vec{x}$\\
\hline

$\theta^j_i$ & parameter in $[0,\epsilon/N]$\\
\hline


$\Omega(\vec{x})$ & distribution based on $\vec{x}$, $\epsilon$ and $\vec{\theta}$\\
\hline






$x_i$ & to accept (1) or reject (0) bid $i$\\
\hline

$\vec{x}^{\ast}$ & optimal allocation solution of ILP (\ref{eqn:socialwelfare})\\
\hline

$\vec{x}^p$& optimal allocation solution of ILP (\ref{eqn:socialwelfare_p})\\
\hline

$\vec{x}^f$ & fractional solution perturbed from $\vec{x}^p$\\
\hline


$\vec{y}^{\epsilon}$ & auction's final allocation solution\\
\hline




$p_i(\vec{y}^{\epsilon})$ & payment of bid $i$ under $\vec{y}^{\epsilon}$\\
\hline


\end{tabular}
\end{table}
\subsection{Goals of Mechanism Design}
We pursue the following properties in our mechanism design. (i) {\em Truthfulness}: The auction mechanism is truthful if for any user $n$, declaring its true valuation of the VM bundle in each of its bids always maximizes its utility, regardless of other users' bids. Truthfulness ensures that selfish buyers are automatically elicited to reveal their true valuations of the VMs they demand, simplifying the bidding strategy and the auction design.
(ii) {\em Social welfare maximization}: The social welfare is the sum of the cloud provider's revenue, $\sum_{w \in  [W]} \sum_{i\in \mathcal{B}_w}p_ix_i$, and the aggregate users' utility $\sum_{w \in  [W]} \sum_{i\in \mathcal{B}_w}(v_i-p_i)x_i$. Since the cloud provider's revenue and the payment from the users cancel out, the social welfare is equivalent to the overall valuation of the winning bids $\sum_{w \in  [W]} \sum_{i\in \mathcal{B}_w}v_ix_i$, which equals $\sum_{w \in  [W]} \sum_{i\in \mathcal{B}_w}b_ix_i$ under truthful bidding. Different from existing work that achieve only approximate social welfare optimality with a ratio much larger than $1$, we seek to achieve $(1-\epsilon)$-optimality where $\epsilon$ is a tunable parameter that can be arbitrarily close to $0$. (iii) {\em Computational efficiency}: A polynomial-time resource allocation algorithm is desirable for the auction to run efficiently in practice. Our auction mechanism leverages the power of randomization to break through the inapproximability barrier of the social welfare maximization problem which does not have a deterministic FPTAS. Consequently, we target polynomial time complexity of the mechanism {\em in expectation}.

Next, we formulate the social welfare maximization problem, which gives rise to the optimal resource allocation solution for the cloud provider to address users' VM demands, assuming truthful bidding is guaranteed. 

{\small
\begin{equation}
{\rm maximize} \hspace*{2mm} \sum_{w \in  [W]} \sum_{i\in \mathcal{B}_w}{b}_ix_i \label{eqn:socialwelfare}
\end{equation}

subject to:\\
\begin{align*}
\sum_{w \in  [W]}\sum_{i\in \mathcal{B}_w} x_i R_{i}^{kd}\leq c_{kd}, &&\forall k\in [K], \forall d \in [D], ~~~~~~~~~(\ref{eqn:socialwelfare}a)\\
\sum_{i\in \mathcal{B}_w}x_i\leq 1, &&\forall w \in  [W], ~~~~~~~~~(\ref{eqn:socialwelfare}b)\\
x_i\in \{0,1\}, &&\forall i\in \mathcal{B}_w,\forall w \in  [W]. ~~~~~(\ref{eqn:socialwelfare}c)\\
\end{align*}
}
\noindent Constraint (\ref{eqn:socialwelfare}a) states that the overall demand for each type of resource in the winning bids should not exceed the overall capacity of the resource in each data center. Constraint (\ref{eqn:socialwelfare}b) specifies that each user can win at most one bid.

\begin{theorem}	\label{thm:nphard}
 The social welfare maximization problem defined in the integer linear program (ILP) (\ref{eqn:socialwelfare}) is NP-hard and 
there does not exist a deterministic FPTAS 
for the problem.
\end{theorem}

The proof is given in \opt{short}{our technical report \cite{FullPaper}}\opt{long}{Appendix \ref{sec:thm:nphard}}.

\section{Auction Design}
\label{sec:auction}

At a high level, our strategy for truthful VM auction design is to apply a randomized VCG-like payment mechanism that works in concert with a randomized allocation algorithm, with the latter achieving optimal social welfare in expectation. Such randomized auctions leverage maximal-in-distributional range (MIDR) algorithms, which are known to be a powerful tool for designing (randomized) truthful mechanisms \cite{journals:siamcomp:DobzinskiD13}. An MIDR algorithm is a randomized allocation algorithm that chooses an allocation solution randomly from a set of feasible solutions of the social welfare maximization problem, following a distribution that is independent of the bidders' bids, and leads to the largest expected social welfare as compared to all other such distributions in a range, {\em e.g.}, the set of distributions over all the feasible solutions.\footnote{To achieve truthfulness of VCG-based mechanisms, there are no additional restrictions on the distributional range or the distributions within the range ({\em e.g.,} the size or specific form) in such MIDR allocation algorithms. In fact, a well-designed range of distributions, as what we will propose, is the key for a better approximation ratio.}
If we can design an MIDR allocation rule, 
then we can combine a randomized VCG payment scheme following a similar distribution to obtain an auction mechanism that is truthful in expectation \cite{journals:siamcomp:DobzinskiD13}. To achieve the other two goals of our auction design,  
 the allocation algorithm should be $(1-\epsilon)$-optimal in social welfare and have polynomial running time in expectation.

%
%

We next establish the randomized allocation algorithm in two steps. First, we design an exact algorithm based on dynamic programming for solving the social welfare maximization problem in Sec. \ref{sec:exactalg}. Next, we design the randomized allocation algorithm based on a perturbation framework, by running the exact algorithm on a randomized perturbed version of the original maximization problem and sampling the final allocation solution from a distribution. The randomized sampling compensates the perturbation on the problem, done before running the exact algorithm, by transforming the optimal solution of the perturbed problem into an near-optimal-in-expectation solution to the original problem.
As the core of the randomized allocation algorithm, the perturbation rule is carefully designed, in order to lead to a $(1-\epsilon)$ approximation ratio, as well as polynomial running time of the algorithm, both in expectation. 

We further describe the payment scheme in Sec.~\ref{sec:payment}. 

\subsection{An Exact Algorithm for Social Welfare Maximization}
\label{sec:exactalg}

The basic idea of the exact algorithm is to enumerate all the feasible allocation solutions excluding those absolutely ``bad'' ones, and then select the optimal allocation solution $\vec{x}=\{x_i,\forall i\in\mathcal{B}_w,w \in [W]\}$ that achieves maximum aggregate bidding price (corresponding to maximum social welfare under truthful bidding) among the set of ``good'' feasible solutions. The set of ``good'' solutions are defined to be those {\em Pareto optimal} solutions which are not dominated by any other feasible solutions, and the ``bad'' ones are those dominated by at least one Pareto optimal solution. This is in line with classical dynamic programming approaches for enumerating Pareto optimal solutions in traditional combinatorial optimization \cite{Nemhauserullmann}.

Let $s(\vec{x})=\sum_{w \in [W]}\sum_{i\in \mathcal{B}_w}{b}_ix_i$ denote the social welfare under allocation solution $\vec{x}$, and $C_{kd}(\vec{x})=\sum_{w \in [W]}\sum_{i\in \mathcal{B}_w}x_iR_{i}^{kd}$ be the total demand for type-$k$ resource in data center $d$ under $\vec{x}$. 
The Pareto optimal solutions are defined as follows.


\begin{definition}\label{def:pareto}(Pareto Optimal Allocation)
An allocation solution $\vec{x}$ is Pareto optimal if it satisfies all the constraints in ILP (\ref{eqn:socialwelfare}), and there does not exist a feasible solution $\vec{x}^{\prime}$ that dominates $\vec{x}$, {\em i.e.}, $\nexists \vec{x}^{\prime}$ such that $s(\vec{x}^{\prime})\ge s(\vec{x})$ and $C_{kd}(\vec{x}^{\prime})\le C_{kd}(\vec{x}),\forall k\in[K],\forall d \in [D]$, with at least one inequality being strict among the above, as well as $\sum_{i\in \mathcal{B}_w}x'_i\le 1$, $\forall w\in[W]$. 
\end{definition}

We identify all the Pareto optimal solutions using a dynamic programming approach: Let $\mathcal{P}(i)$ be the set of all Pareto optimal solutions when we only consider the first $i$ bids in set $[N]$ (the bids in $[N]$ are ordered in any fashion). Let $i$-dimensional vector $\vec{x}^{(i)}$ denote a Pareto optimal solution in $\mathcal{P}(i)$. We compute $\mathcal{P}(i)$ from $\mathcal{P}(i-1)$, and eventually obtain $\mathcal{P}(N)$ which is the set of Pareto optimal solutions of ILP (\ref{eqn:socialwelfare}).

We show the following property of the Pareto optimal solution sets\opt{short}{:}\opt{long}{, with proof given in Appendix \ref{sec:lem:recursion}.}

\begin{lemma}\label{lem:recursion}
If $\vec{x}^{(i)}$ is a Pareto optimal solution in $\mathcal{P}(i)$, then the vector obtained by removing the last element $x_{i}^{(i)}$ from $\vec{x}^{(i)}$ is a Pareto optimal solution in $\mathcal{P}(i-1)$, 
$\forall i=2,\ldots,N$.
\end{lemma}
\noindent Let $\mathcal{P}(i-1)+1$ denote the set of $i$-dimensional solutions obtained by simply adding $1$ as the $i$th element to each solution vector in $\mathcal{P}(i-1)$ (removing infeasible solutions), and $\mathcal{P}(i-1)+0$ be the set obtained by adding $0$ as the $i$th element. Given Lemma \ref{lem:recursion}, we know that any solution in $\mathcal{P}(i)$ must be contained in set $\mathcal{P}(i-1)+0\cup\mathcal{P}(i-1)+1$. In the algorithm given in Alg.~\ref{alg:pareto}, we start with $\mathcal{P}(1)$, which contains two Pareto optimal solutions $1$ (accept $B_1$) and $0$ (reject $B_1$), if the resource demands in bid $B_1$ do not exceed the respective capacity limits, and contains only one Pareto optimal solution $0$, otherwise. Then we construct $\mathcal{P}(i), i=2,\ldots, N$, by eliminating infeasible or non-Pareto-optimal solutions from $\mathcal{P}(i-1)+0\cup\mathcal{P}(i-1)+1$. Finally, the exact allocation solution of ILP (\ref{eqn:socialwelfare}) is obtained as the solution in $\mathcal{P}(N)$ that achieves the maximum social welfare.

\begin{algorithm}[!t]
	{\label{alg:pareto}}
    \caption{The Exact Algorithm for ILP (\ref{eqn:socialwelfare})}
    \label{alg:pareto}

    \SetKwFunction{StoogeSort}{StoogeSort}
    \SetKwFunction{Swap}{Swap}

    \SetKwData{OneThird}{OneThird}
    \SetKwData{TwoThirds}{TwoThirds}

    \SetArgSty{}
    	{\bf Input:} $\vec{b}, \vec{R}, \vec{c}$\\
		{\bf Output:} exact optimal solution $\vec{x}$\\
		
    	\If {$C_{kd}(\{1\}) \le c_{kd},\forall k\in [K], \forall d\in [D]$}{
		$\mathcal{P}(1)=\{0,1\}$;
		}
		\Else {$\mathcal{P}(1)=\{0\}$;}
		
  		\For {$i=2,..., N$}{
        \For {all $\vec{x}^{(i-1)}\in \mathcal{P}(i-1)$}{
		$\vec{x}^{(i)}=\{\vec{x}^{(i-1)},1\}$;\\
        \If {$\vec{x}^{(i)}$ satisfies Constraints (\ref{eqn:socialwelfare}a) and (\ref{eqn:socialwelfare}b)} {
        Put $\vec{x}^{(i)}$ into $\mathcal{P}(i-1)+1$\;
          }
        }
        Merge $\mathcal{P}(i-1)+0$ and $\mathcal{P}(i-1)+1$ into $\mathcal{P}(i)^{\prime}$\;
        Prune the solutions dominated by others in $\mathcal{P}(i)^{\prime}$ to obtain $\mathcal{P}(i)=\{\vec{x}^{(i)}\in \mathcal{P}(i)^{\prime}|\nexists \vec{x}^{(i)\prime}\in \mathcal{P}(i)^{\prime}:\vec{x}^{(i)\prime}~\mbox{dominates}~\vec{x}^{(i)}\}$\;

  }
  
\Return $\vec{x}=\mbox{argmax}_{\vec{y} \in \mathcal{P}(N)}s(\vec{y})$
\end {algorithm} 

The computation complexity of the exact algorithm in Alg.~\ref{alg:pareto} is polynomial in the number of Pareto optimal solutions in $\mathcal{P}(N)$, as given in Theorem \ref{thm:boundedtime}, which is based on Lemma \ref{nondecreasingparetoset}. 

\begin{lemma}\label{nondecreasingparetoset}
The number of Pareto optimal solutions $|\mathcal{P}(i)|$ does not decrease with $i$, i.e., $|\mathcal{P}(1)| \le ... \le |\mathcal{P}(N)|$.
\end{lemma}
The proof is given in \opt{short}{our technical report \cite{FullPaper}}\opt{long}{Appendix \ref{sec:lem:nondecreparetoset}}.

\begin{theorem}\label{thm:boundedtime}
The computation complexity of Alg.~\ref{alg:pareto} is $O(KDN|\mathcal{P}(N)|^2)$.
\end{theorem}
The proof is given in \opt{short}{our technical report \cite{FullPaper}}\opt{long}{Appendix \ref{sec:thm:boundedtime}}.

The algorithm runs in exponential time in the worst case, since there can be exponentially many Pareto optimal solutions to check in the worst case. In what follows, however, we will show that this exact algorithm is efficient in practice, {\em i.e.}, running in polynomial time in expectation, and can be used as a building block in a perturbation framework for producing a randomized allocation algorithm. 

\subsection{The Randomized $(1-\epsilon)$-Approx. Allocation Algorithm}
\label{sec:randalg}

We next design the randomized algorithm to solve the social welfare maximization problem in (\ref{eqn:socialwelfare}) in polynomial time in expectation. The basic idea is to obtain a set of feasible allocation solutions that achieve $(1-\epsilon)$-optimal social welfare in expectation, following a well-designed distribution, and then randomly output an allocation solution from this set following this distribution. To achieve computation efficiency, the set of feasible solutions are to be computed in polynomial time in expectation, including one solution from solving the random perturbation of the social welfare maximization problem, based on smoothed analysis techniques \cite{Spielman:2001:SAA:380752.380813}\cite{Roglin:2009ku}. The random perturbation on the original problem is carefully designed, in close connection with the distribution to sample feasible solutions, to achieve $(1-\epsilon)$-optimal social welfare of (\ref{eqn:socialwelfare}) in expectation. Especially, the most salient feature of algorithm design in this work, as the first in the literature, is to apply a pair of associated random perturbation schemes for smoothed polynomial time algorithm design and for randomized auction design, respectively.


\noindent {\bf Algorithm design.}~Given an arbitrary parameter $\epsilon\in (0,1)$ and $KDN$ random variables $\{\theta^j_1, \theta^j_2, \ldots, \theta^j_N\}_{j\in \{0, \ldots, KD-1\}}$ that are independently and identically chosen from a uniform distribution on the interval of $[0,\frac{\epsilon}{N}]$. 
Let $\vec{\theta}^j=\{\theta^j_1,\ldots,\theta^j_N\}, \forall j\in \{0, \ldots, KD-1\}$. 
Suppose the packing constraints in (\ref{eqn:socialwelfare}a) are ordered. 
 Let $j\in [KD]$ index the sorted constraints; thus we can use ${R}^j_i$ to replace ${R}^{kd}_i$ in (\ref{eqn:socialwelfare}a), $\forall j\in \{1,\ldots,KD\}$, and will refer to $j$ as a resource, which represents the corresponding resource ($k$) in the respective data center ($d$).

We perturb the bidding price $b_i$ in the objective function and the demand $R^{i}_{kd}$ of the first $K\times D -1$ packing constraints of ILP (\ref{eqn:socialwelfare}) independently to:
\begin{align}
&\hat{b}_i=(1-\epsilon/2){b}_i+\frac{\theta^0_i\sum_{i^{\prime}=1}^N{{b}_{i^{\prime}}}}{N}, \forall i\in[N]~~\label{eqn:perturb_b};\\
&\hat{R}^j_i={R}^j_i+\frac{\theta^j_i\sum_{i^{\prime}=1}^N{{R}^j_{i^{\prime}}}}{N}, \forall i\in[N], j\in\{1, \ldots, KD-1\}~~\label{eqn:perturb_R}
\end{align}

\noindent Here, $\theta^0_i, \forall i\in[N]$, are the random variables associated with $b_i$'s, the coefficients in the objective function, and $\theta^j_i, \forall i\in[N]$, are associated with ${R}^j_i$'s in the $j$th constraint in (\ref{eqn:socialwelfare}a), $\forall j\in\{1, \ldots, KD-1\}$. Note that the last constraint in (\ref{eqn:socialwelfare}a) is not perturbed. We define $\hat{R}^{KD}_i = R^{KD}_i, \forall i\in[N]$, for this unperturbed last constraint in (\ref{eqn:socialwelfare}a).\footnote{In line with the latest smoothed analysis techniques, we adopt the semi-random model \cite{Spielman:2001:SAA:380752.380813}. 
In this context, for a binary maximization problem with one objective function and multiple packing constraints, the objective function and all the packing constraints except the last one are perturbed due to the reasons explained in Appendix \ref{sec:thm:expectednumber}.} 
The perturbed social welfare maximization problem is:

{\small
\begin{equation}
{\rm maximize} \hspace*{2mm} \sum_{w \in [W]} \sum_{i\in \mathcal{B}_w} \hat{b}_i x_i \label{eqn:socialwelfare_p}
\end{equation}

subject to:\\
\begin{align*}
&\sum_{w \in [W]}\sum_{i\in \mathcal{B}_w} x_i \hat{R}_i^j\le c_j, \forall j\in \{1, \ldots, KD\}, ~~~~~~~~~(\ref{eqn:socialwelfare_p}a)\\
&(\ref{eqn:socialwelfare}b)~\text{and}~(\ref{eqn:socialwelfare}c)
\end{align*}
}

\vspace{0.5mm}


\noindent According to the perturbation in \eqref{eqn:perturb_b}, let
{\small 
\begin{equation}\label{eqn:perturbmatrix}
P=(1-\epsilon/2)I+\frac{\vec{\theta^0}\vec{1}^T}{N}
\end{equation}
}

\noindent be the perturbation matrix of the objective function, where $I$ is the $N\times N$ identity matrix. Then we can express the perturbation of bidding price vector as $\vec{\hat{b}}=P\vec{b}$. 
We solve the perturbed social welfare maximization problem using the exact algorithm (Alg.~\ref{alg:pareto}), and derive the optimal solution $\vec{x}^{p}$ and optimum value of the perturbed objective function $POPT=\vec{\hat{b}}^T \vec{x}^p$. We will show that the expected running time to solve the randomly perturbed ILP is polynomial in Theorem \ref{thm:expectednumber} and Theorem
\ref{thm:expectedtime}. 



Let $\vec{x}^{\ast}$ be the optimal solution of ILP (\ref{eqn:socialwelfare}), and $OPT=\vec{b}^T\vec{x}^{\ast}$ be the optimal social welfare. The following lemma shows that the optimal objective value of the perturbed problem is at least $(1-\epsilon)$-fraction of the optimal social welfare of the original problem, which is very close as long as the perturbation, decided by $\epsilon$, is small enough, under a small bid assumption.  
The proof is given in Appendix \ref{sec:lem:POPT_vs_OPT}.

\begin{lemma}\label{lem:POPT_vs_OPT}
$POPT \ge (1-\epsilon) OPT$, 
 if $\max_{i\in [N], k\in[K], d\in[D]}\frac{R^{kd}_i}{c_{kd}} \le \frac{1}{2KD(2+\frac{1}{\epsilon})}$.
\end{lemma}
The small bid assumption stated in Lemma \ref{lem:POPT_vs_OPT} 
essentially requires that the demand in each bid for each type of resource in each desired datacenter is small as compared to the corresponding resource capacity, which is easy to justify in real systems. Moreover, if a smaller $\epsilon$ is chosen, the assumption becomes stronger, {\em i.e.,} the ratio of the largest demand among the bids for each type of resource over the resource capacity is required to be smaller.

Lemma \ref{lem:POPT_vs_OPT} gives that $POPT=(P\vec{b})^T\vec{x}^p=\vec{b}^T(P^T\vec{x}^p)\ge (1-\epsilon) OPT$. We can obtain a potential solution $\vec{x}^f=P^T\vec{x}^p$ to the original problem,
 which achieves $(1-\epsilon)$-optimal social welfare. 
However, the bad news is that $\vec{x}^f$ may well be fractional due to the fractional entries in $P^T$, and hence not a feasible solution of ILP (\ref{eqn:socialwelfare}) (not to mention whether it satisfies other constraints in (\ref{eqn:socialwelfare}) or not). We hence cannot directly use $\vec{x}^f$ as the allocation solution to our social welfare maximization problem (\ref{eqn:socialwelfare}), but design a random sampling approach to produce a feasible allocation solution from a set of feasible solutions of (\ref{eqn:socialwelfare}) following a well-designed distribution, such that the expectation of the randomly produced solution is $\vec{x}^f$, which achieves $(1-\epsilon)$-optimal social welfare in expectation. 

Let $\vec{l}_i$ denote a solution of (\ref{eqn:socialwelfare}) that accepts only the $i$th bid and rejects all the other bids, {\em i.e.}, $l_i^i=1$ and $l_i^{i'}=0,\forall i'\ne i$. We can easily see that $\vec{l}_i,\forall i\in N$, are feasible solutions to (\ref{eqn:socialwelfare}). Note that $\vec{x}^p$ is a feasible solution to (\ref{eqn:socialwelfare}) as well, since all the perturbed coefficients of the packing constraints in ILP (\ref{eqn:socialwelfare_p}) are no smaller than those of ILP (\ref{eqn:socialwelfare}). The set of feasible solutions to sample from hence is $\{\vec{x}^p,\vec{l}_1, \ldots, \vec{l}_N, \vec{0}\}$, where $\vec{0}$ is a $N$-dimensional all-zero vector. The final allocation solution of  (\ref{eqn:socialwelfare}), denoted by $\vec{y}^{\epsilon}$, is randomly produced following the distribution $\Omega(\vec{x}^p)$ below: 
\begin{equation}\label{eqn:optimaldistribution}
{\small
\Omega(\vec{x}^p)= \left\{ 
\begin{array}{l l l l}
Pr[\vec{y}^{\epsilon}=\vec{x}^{p}]=1-\epsilon/2,  \\
Pr[\vec{y}^{\epsilon}=\vec{l}_i]=\frac{\sum_{j=1}^N\theta^0_j {x}^p_j}{N}, \forall i\in \{1,...,N\},\\
Pr[\vec{y}^{\epsilon}=\vec{0}]
=1-Pr[\vec{y}^{\epsilon}=\vec{x}^{p}]-\sum_{i=1}^N Pr[\vec{y}^{\epsilon}=\vec{l}_i].\\
\end{array} \right.}
\end{equation}

\noindent We can verify that the probabilities of all candidate solutions are positive and sum up to exactly $1$. We then have that the expectation of $\vec{y}^{\epsilon}$ is

\begin{equation}
{\small
E[\vec{y}^{\epsilon}]=(1-\epsilon/2)\vec{x}^p+(\frac{\sum_{j=1}^N\theta^0_j{x}^p_j}{N})(\sum_{i=1}^N\vec{l}_i)=P^T\vec{x}^p=\vec{x}^f. \label{eqn:expect_y}}
\end{equation}

\noindent Given the above, the design of all the candidate solutions in $\Omega(\vec{p})$ and probability assignment of each of them were aiming to make the expectation equal to $P^T\vec{x}^p$.
The high level idea is using $\Omega(\vec{x}^p)$ to randomly perturb $\vec{x}^p$ to $\vec{x}^f$ where $P^T$ of $\vec{x}^p$ compensates the perturbation $P$ of $\vec{b}$ in \eqref{eqn:perturbmatrix}. According to the critical property $(P\vec{b})^T\vec{x}^p=\vec{b}^T(P^T\vec{x}^p)$, i.e., the perturbation of the objective function is equal to the perturbation of the solution, $\vec{x}^f$ enables the $(1-\epsilon)$ approximation. 
Here, $\vec{y}^{\epsilon}$ is equal to $x^p$ with a high probability of $1-\epsilon/2$, which is in accordance with the $1-\epsilon/2$ part in \eqref{eqn:perturbmatrix}. Each of the base vectors $\vec{l}_i$ and the zero vector $\vec{0}$ is chosen as a candidate to make $\Omega(\vec{x}^p)$ diffuse enough, such that an expected polynomial number of Pareto optimal solutions to \eqref{eqn:socialwelfare_p} can be guaranteed, which will be proved in Theorem \ref{thm:expectednumber}.

\noindent\textbf{Algorithm steps.} We summarize the above steps in Alg.~\ref{alg:perturbation}, which is our randomized algorithm for computing a $(1-\epsilon)$-approximate solution to social welfare optimization problem (\ref{eqn:socialwelfare}).

\begin{algorithm}
\caption{The $(1-\epsilon)$-Approx.~Algorithm for ILP (\ref{eqn:socialwelfare})}
\label{alg:perturbation}

\SetKwFunction{StoogeSort}{StoogeSort}
\SetKwFunction{Swap}{Swap}
\SetKwData{OneThird}{OneThird}
\SetKwData{TwoThirds}{TwoThirds}
\SetArgSty{}
	{\bf Input:} $\epsilon\in (0,1), \vec{b}, \vec{R}, \vec{c}$\\
	{\bf Output:} $(1-\epsilon)$-approximate allocation solution $\vec{y}^{\epsilon}$\\	
	 Choose $\{\theta^j_1,...,\theta^j_N\}_{j\in\{0,\ldots, KD-1\}}$ independently and identically in the interval $[0,\frac{\epsilon}{N}]$\;
	 Construct each perturbed parameter $\hat{b}_i$ and $\hat{R}^{kd}_{i}$ according to \eqref{eqn:perturb_b} and \eqref{eqn:perturb_R}, respectively\;
	 Compute $\vec{x}^{p}=$ Algorithm \ref{alg:pareto}$(\vec{\hat{b}}, \mathbf{\hat{R}},\mathbf{c})$\;
	 Produce distribution $\Omega(\vec{x}^p)$ according to (\ref{eqn:optimaldistribution})\;
	\Return A sample $\vec{y}^{\epsilon}$ according to $\Omega(\vec{x}^p)$ in (\ref{eqn:optimaldistribution}).
\end{algorithm}
%
%

We further illustrate our algorithm using a simple example in Fig.~\ref{illustration:algorithm}. We consider a toy example where only two bids are submitted to the auctioneer. Suppose (0,0), (0,1), (1,1) and (1,0) are the feasible solutions of the original problem in \eqref{eqn:socialwelfare}. Since the coefficients of the packing constraints $\hat{R}^{kd}_i$ are randomly enlarged to be $\hat{R}^{kd}_i$ in the perturbed problem \eqref{eqn:socialwelfare_p}, the set of feasible solutions of the perturbed problem could be shrunk. An extreme case could happen where both bids are accepted in the optimal solution of \eqref{eqn:socialwelfare}, {\em i.e.}, $\vec{x}^{\ast}=(1,1)$, while at most one of the two bids could be accepted in the optimal solution of \eqref{eqn:socialwelfare_p} (the green area denotes the feasible region of the perturbed problem). By running the exact algorithm Alg. \ref{alg:pareto} to solve \eqref{eqn:socialwelfare_p}, the optimal solution $\vec{x}^p=(1,0)$ is obtained.
The near-optimality of Alg. \ref{alg:perturbation} is basically achieved by two facts. 
One is that due to the small perturbation of $\hat{R}^{kd}_i$'s, there is an extremely small probability\footnote{For rigorousness, this probability could be quite high when there are a large number of bids submitted. 
Fortunately, according to the proof of Lemma \ref{lem:POPT_vs_OPT}, with the small bid assumption, even if the optimal solution of the original problem is infeasible to the perturbed problem, the objective value of the perturbed problem under $\vec{x}^{p}$ will not deviate much to that under $\vec{x}^{\ast}$.} that the optimal solution of the original problem is infeasible to the perturbed problem.
The other is that the perturbation of the vector $\vec{b}^T$ could be compensated by the perturbation of the feasible solutions of the perturbed problem. As Fig.~\ref{illustration:algorithm} (a) shows, Alg. \ref{alg:perturbation} first perturbs $\vec{b}$ to $P\vec{b}$ by rotating $\vec{b}$ by a small angle $\alpha$, and exactly solves for the maximal $(P\vec{b})^T\vec{x}$ from the three feasible solutions (blue points in Fig. \ref{illustration:algorithm}(a)). It is equivalent to rotating each feasible solution of the perturbed problem by an angle of $- \alpha$ and solving for the maximal $\vec{b}^T(P^T\vec{x})$ from the rotated solutions $P^T\vec{x}$'s (pink points in Fig. \ref{illustration:algorithm}(b)). Directed by this insight, we randomly choose $\vec{y}^{\epsilon}$ from a set of feasible solutions of the original problem with the expectation of $P^T\vec{x}^p$. Because $\vec{b}^T(P^T\vec{x}^p)=(P\vec{b})^T\vec{x}$, the expected social welfare gained by Alg. \ref{alg:perturbation}, which is $\vec{b}^T\vec{y}^{\epsilon}$, preserves the approximation ratio of $(P\vec{b})^T\vec{x}^p$ to $\vec{b}^T\vec{x}^{\ast}$. 
\begin{figure}[!t]
\begin{center}
  \includegraphics[width=0.48\textwidth]{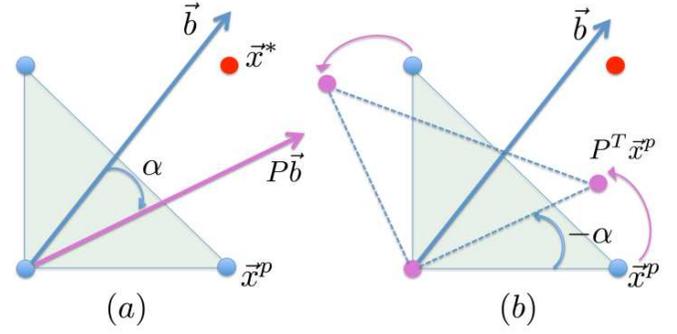}
  \caption{An example to show the mechanism design}
  \label{illustration:algorithm}
\end{center}
\end{figure}

\noindent\textbf{Analysis.} Alg.~\ref{alg:perturbation} achieves the following properties.

\noindent {\bf (i) The expected running time of the randomized Alg.~\ref{alg:perturbation} is polynomial.}  Although the worst-case computation complexity of the exact Algorithm in Alg.~\ref{alg:pareto} is exponential due to exponentially many Pareto optimal solutions in the worst case (Theorem \ref{thm:boundedtime}), we show that the algorithm runs efficiently in practice, based on smoothed analysis techniques \cite{Spielman:2001:SAA:380752.380813}\cite{Roglin:2009ku}. The reason is that the expected number of the Pareto optimal solutions
of the perturbed social welfare maximization problem in (\ref{eqn:socialwelfare_p}) is polynomial, and hence the exact algorithm runs in polynomial time in expectation when applied to the perturbed problem---perturbed with a $P$ generated randomly as in (\ref{eqn:perturbmatrix}). According to smoothed analysis, Alg.~\ref{alg:pareto} is said to run in {\em smoothed polynomial time}.

\begin{theorem}\label{thm:expectednumber}
The expectation of the random variable $|\mathcal{P}(N)|^2$ of the perturbed social welfare maximization problem (\ref{eqn:socialwelfare_p}) is upper bounded by $O(N^{8KD}/\epsilon^{2KD})$, where the perturbed parameters are produced according to \eqref{eqn:perturb_b} and \eqref{eqn:perturb_R} with $\{\theta^j_1, \theta^j_2, \ldots, \theta^j_N\}_{j\in\{0,..., KD-1\}}$ independently and identically chosen from a uniform distribution on the interval of $[0,\frac{\epsilon}{N}]$.
\end{theorem}
\vspace{-1mm}
The proof is given in Appendix \ref{sec:thm:expectednumber}.


\begin{theorem}\label{thm:expectedtime}
The expected running time of the randomized algorithm Alg.~\ref{alg:perturbation} is polynomial.\\
\end{theorem}
\vspace{-5mm}
The proof is given in Appendix \ref{sec:thm:expectedtime}.\\

\noindent {\bf (ii) Alg.~\ref{alg:perturbation} achieves $(1-\epsilon)$-optimal social welfare.}

\begin{theorem}	\label{thm:smoothedfptas}
Alg.~\ref{alg:perturbation} is a $(1-\epsilon)$-approximation randomized algorithm for the social welfare maximization problem
(\ref{eqn:socialwelfare}), under the small bid assumption $\max_{i\in [N], k\in[K], d\in[D]}\frac{R^{kd}_i}{c_{kd}} \le \frac{1}{2KD(2+\frac{1}{\epsilon})}$.
\end{theorem}

\begin{proof} 
We have shown $E[\vec{y}^{\epsilon}]=\vec{x}^f$. Hence $E[\vec{b}^T\vec{y}^{\epsilon}]=\vec{b}^T\vec{x}^f=\vec{b}^T(P^T\vec{x}^{p})=(P\vec{b})^T\vec{x}^p
=POPT\ge (1-\epsilon)OPT$, based on Lemma \ref{lem:POPT_vs_OPT}.

\end{proof}

\vspace{-5mm}
\subsection{The Truthful-in-Expectation VM Auction}
\label{sec:payment}
Recall the MIDR mechanism design \cite{journals:siamcomp:DobzinskiD13} that we introduced at the beginning of this section. By now we have designed the randomized allocation algorithm which chooses an allocation solution following the distribution $\Omega(\vec{x}^p)$ in (\ref{eqn:optimaldistribution}). This distribution is independent of the cloud users' bids. We next show that it leads to the largest expected social welfare among a compact set of such distributions, such that we can combine a randomized VCG payment scheme following a similar distribution, to obtain an auction mechanism that is truthful in expectation.


\begin{theorem}
The randomized allocation Alg.~\ref{alg:perturbation} is an MIDR allocation rule for the social welfare maximization problem
(\ref{eqn:socialwelfare}).
\end{theorem}
\begin{proof}
An MIDR allocation rule of a social welfare maximization problem returns an allocation solution that is sampled randomly from a distribution over a feasible set of the problem, which achieves the largest expected social welfare, among random solutions produced following distributions in a distributional range, which is a fixed compact set of probability distributions over the feasible set that are independent of the users' bids \cite{journals:siamcomp:DobzinskiD13}.

Let $\mathcal{T}$ and $\mathcal{\hat{T}}$ denote the set of all feasible solutions of ILP (\ref{eqn:socialwelfare}) and of the perturbed ILP \eqref{eqn:socialwelfare_p}, respectively.
Recall the perturbation rule for each demand parameter in \eqref{eqn:perturb_R}. We randomly enlarge the coefficient of the packing constraints from $R^{kd}_i$ to $\hat{R}^{kd}_i$, therefore, $\mathcal{\hat{T}}$ is a subset of $\mathcal{T}$. 
For each feasible solution $\vec{x}\in \mathcal{\hat{T}}$, we can obtain a distribution $\Omega(\vec{x})$ in the same way as the distribution in (\ref{eqn:optimaldistribution}) by replacing all the $\vec{x}^{p}$ with $\vec{x}$ shown in $\Omega(\vec{x}^p)$. Then let $\vec{y}$ denote the random allocation solution produced following $\Omega(\vec{x})$, i.e., $\vec{y}\sim\Omega(\vec{x})$.
Given $\epsilon$ and $\{\theta^j_1,\ldots,\theta^j_N\}_{j\in\{0, \ldots, KD-1\}}$, the distribution $\Omega(\vec{x})$ is dependent on feasible solution $\vec{x}$, 
but independent of the users' bids. $\mathcal{R}=\{\Omega(\vec{x}),\forall \vec{x}\in \mathcal{\hat{T}}\}$ is a compact set including all the distributions indexed by feasible solution $\vec{x}$ in the set $\mathcal{\hat{T}}$. Due to $\mathcal{\hat{T}}\subset \mathcal{T}$, $\mathcal{R}$ is also a compact set of distributions over the feasible solutions in $\mathcal{T}$. 
Using $\mathcal{R}$ as the distributional range, we have

\vspace{-5mm}
{\small
\begin{eqnarray}
E_{\vec{y}^{\epsilon}\sim \Omega(\vec{x}^p)}[\vec{b}^T\vec{y}^{\epsilon}]&=&(P\vec{b})^T\vec{x}^p = \max_{\vec{x}\in\mathcal{\hat{T}}}(P\vec{b})^T\vec{x} \nonumber\\
&=&\max_{\vec{x}\in \mathcal{\hat{T}}}\vec{b}^T(P^T\vec{x})\nonumber\\
&=&\max_{\Omega(\vec{x}) \in \mathcal{R}}E_{\vec{y}\sim \Omega(\vec{x})}[\vec{b}^T\vec{y}].\label{eqn:e_by}
\end{eqnarray}
}
\vspace{-3mm}

\noindent The first equation is due to (\ref{eqn:expect_y}). The second equation is because $\vec{x}^p$ is the optimal solution of the perturbed ILP (\ref{eqn:socialwelfare_p}). 
The last equation is due to $E_{\vec{y}\sim \Omega(\vec{x})}[\vec{y}]=P^T\vec{x}$, which can be readily obtained according to \eqref{eqn:expect_y}. Hence the solution $\vec{y}^{\epsilon}$ selected following distribution $\Omega(\vec{x}^p)$ in Alg.~\ref{alg:perturbation} achieves the largest expected social welfare, among all the solutions produced following distributions in the distributional range $\mathcal{R}$, leading to an MIDR allocation rule. 
\end{proof}

\vspace{-0.5mm}
Now we describe our randomized VCG payment, which is based on the allocation solution $\vec{y}^{\epsilon}$, as follows:
{\small
\begin{equation}
p_i(\vec{y}^{\epsilon})=\vec{b}_{-i}^T\vec{y}^{\epsilon}_{-i}-(\vec{b}^T\vec{y}^{\epsilon}-b_i{y}^{\epsilon}_i),\forall i\in[N].\label{eqn:payment}
\end{equation}
}
\noindent Here $\vec{b}_{-i}$ denotes the bidding price vector where the $i$th bidding price is set to $0$; $\vec{y}^{\epsilon}_{-i}$ is the random allocation solution output by Alg.~\ref{alg:perturbation} with the input bidding price vector $\vec{b}_{-i}$.
Hence $\vec{b}_{-i}^T\vec{y}^{\epsilon}_{-i}$ is the social welfare when the $i$th bid is excluded from the auction. Further recall $\vec{y}^{\epsilon}$ is the output of Alg.~\ref{alg:perturbation} with the full bidding price vector $\vec{b}$. Let ${y}^{\epsilon}_i$ be the $i$th element of $\vec{y}^{\epsilon}$. Hence $\vec{b}^T\vec{y}^{\epsilon}-b_i{y}^{\epsilon}_i$ is the social welfare achieved by all the other bids except bid $i$, when all the bids are considered in the auction.

\begin{algorithm}[!t]
\caption{The Randomized Auction Mechanism}
\label{alg:auction}
	{\bf Input:} $\epsilon\in (0,1), \vec{b}, \vec{R}, \vec{c}$\\
	{\bf Output:} allocation solution $\vec{y}^{\epsilon}$ and payment $\vec{p}$\\

	Compute $\vec{y}^{\epsilon}$=Algorithm \ref{alg:perturbation}$(\epsilon, \vec{b}, \vec{R}, \vec{c})$\;
	Compute payment $p_i(\vec{y}^{\epsilon})=\vec{b}_{-i}^T\vec{y}^{\epsilon}_{-i}-(\vec{b}^T\vec{y}^{\epsilon}-b_i{y}^{\epsilon}_i)$, for all accepted bids $i\in [N]$\;
	\Return $\vec{y}^{\epsilon}$ and $\vec{p}$
\end{algorithm}

We summarize our complete randomized auction mechanism in Alg.~\ref{alg:auction}. The following theorem shows that the auction design fulfils our design objectives.

\begin{theorem}\label{thm:truthful}
The auction mechanism in Alg.~\ref{alg:auction}, which combines the randomized allocation algorithm in Alg.~\ref{alg:perturbation} and the randomized VCG payment in (\ref{eqn:payment}), runs in polynomial time in expectation, is truthful in expectation, and achieves $(1-\epsilon)$-optimal social welfare in expectation, under the small bid assumption $\max_{i\in [N], k\in[K], d\in[D]}\frac{R^{kd}_i}{c_{kd}} \le \frac{1}{2KD(2+\frac{1}{\epsilon})}$.
\end{theorem}

\begin{proof}
According to the principles of MIDR algorithms \cite{journals:siamcomp:DobzinskiD13}, to render a truthful-in-expectation mechanism, we should combine an MIDR allocation rule with a VCG-like payment as follows:

\vspace{-5mm}
\begin{equation}\label{vcglike}
{\small
p_i^{\prime}=E[\vec{b}_{-i}^T\vec{y}^{\epsilon}_{-i}-(\vec{b}^T\vec{y}^{\epsilon}-b_i{y}^{\epsilon}_i)],\\	}			
\end{equation}

\noindent where the expectation is computed as follows in more details:
$E_{\vec{y}^{\epsilon}_{-i}\sim \Omega(\vec{x}^p_{-i})}[\vec{b}_{-i}^T\vec{y}^{\epsilon}_{-i}]-E_{\vec{y}^{\epsilon}\sim \Omega(\vec{x}^p)}[\vec{b}^T\vec{y}^{\epsilon}-b_i{y}^{\epsilon}_i]$. Here $\vec{x}^p_{-i}$ is the optimal solution of the perturbed ILP (\ref{eqn:socialwelfare_p}), produced by line 5 of Alg.~\ref{alg:perturbation}, when the input bidding price vector is $\vec{b}_{-i}$.


	Since it may not be always possible to compute the expectation in (\ref{vcglike}) efficiently, it has been proved \cite{conf:focs:DughmiR10} that instead of using (\ref{vcglike}), we can use a randomized payment rule to yield the truthfulness in expectation as well, as long as the expected payment of the randomized payment rule equals $p^{\prime}_i$ in (\ref{vcglike}).





The expectation of our random payment in (\ref{eqn:payment}) is exactly 

\vspace{-5mm}
\begin{eqnarray}
	&&E[p_i(\vec{y}^{\epsilon})]\nonumber\\
	&=&E_{\vec{y}^{\epsilon}_{-i}\sim \Omega(\vec{x}^p_{-i})}[\vec{b}_{-i}^T\vec{y}^{\epsilon}_{-i}]-E_{\vec{y}^{\epsilon}\sim \Omega(\vec{x}^p)}[\vec{b}^T\vec{y}^{\epsilon}-b_i{y}^{\epsilon}_i]\nonumber\\
	&=&p^{\prime}_i.\nonumber
\end{eqnarray}	
Hence the random payment in (\ref{eqn:payment}) renders truthfulness in expectation and can be computed in polynomial time in expectation. Combining Theorem \ref{thm:expectedtime} and Theorem \ref{thm:smoothedfptas}, this theorem is proved.
\end{proof}


\section{Performance Evaluation}
\label{sec:simulation}

We evaluate our randomized auction using trace-driven simulations, exploiting Google cluster-usage data \cite{clusterdata:Reiss2011} which record jobs submitted to the Google cluster. Each job contains multiple tasks, with information on resource demands (CPU, RAM, Disk) of the tasks. We translate each job into a VM bundle bid, and the resource demand of each task in the job into a VM in the bundle. 
There are around $1000$ types of VMs consisting of $K=3$ types of resources in our experiments.

%
Each task is mapped to a VM of a specific type by resource demands, and further mapped to a data centre randomly. We then generate the VM demands of the bundle, $q_{md}^i$, by counting the number of VMs of the same type mapped to the same data center among the tasks in the job. We estimate the unit prices of CPU, Disk and RAM, respectively, based on the prices of Amazon EC2 instances and their resource composition, and then set the bid price of each bundle based on the unit prices and the overall resource demands in the bundle, scaled by a random number in $[0.75, 1.5]$. In this way, we obtain a pool of bidding bundles from the Google cluster data.
Each user randomly chooses at most $|\mathcal{B}_w|$ bundles from the pool to bid. 
We compute the capacity of type-$k$ resource, $c_{dk}$, in a data center based on the overall amount of this resource required in this data center in all the bid bundles submitted by the users, and scale it down using a random factor in [0, 0.5W/N], 
such that roughly no more than half of the users can win a bid under constraint (\ref{eqn:socialwelfare}b), without loss of generality. 
By default, the number of users is $W=500$, the upper-bound on the number of bids a user can submit is $|\mathcal{B}_w|=4$, the number of data centers is $D=8$, and $\epsilon=0.05$. We repeat each experiment for 50 times to obtain the average results. 
Note that some of the resource demands ($R^{kd}_i$) violate the small bid assumption that is required in the analysis of the approximation ratio of our allocation algorithm in Lemma \ref{lem:POPT_vs_OPT}, Theorem \ref{thm:smoothedfptas} and Theorem \ref{thm:truthful}. In this sense, our algorithm could work well even when the bid demands are not as small as we require in the theoretical analysis.

\subsection{Approximation Ratio}
We first study the average approximation ratio achieved by our algorithm, computed by the social welfare achieved by Alg.~\ref{alg:perturbation} over the optimal social welfare by solving (\ref{eqn:socialwelfare}) exactly. Let RPAA represent our Randomized Perturbation-based Approximation Algorithm in Alg.~\ref{alg:perturbation} in the figures. Given $\epsilon=0.05$, the theoretical expected approximation ratio is $0.95$. Fig.s~\ref{fig:comp1}-\ref{fig:comp2} show that the average approximation ratio is larger than the theoretical ratio and approaches the latter when the number of users (or bids) increases. According to Lemma \ref{lem:POPT_vs_OPT} 
 and (\ref{eqn:expect_y}) in Sec.~\ref{sec:auction}, $1-\epsilon$ is a lower bound of the approximation ratio, and the results show that under practical settings the average ratio achieved is better. The reason why the ratio approaches the theoretical one when $W$ is large can also be explained by (\ref{eqn:POPT_OPT_1}) in the proof of Lemma \ref{lem:POPT_vs_OPT}, that the inequality tends to equality when $W$ (and hence $N$) is large. 
Fig. \ref{fig:comp1} and Fig. \ref{fig_k} also demonstrate that the average ratio decreases slightly with the increase of the number of data centers or the number of resource types, respectively. The reasons are the same: either a larger $D$ or a large $K$ leads to more constraints in (\ref{eqn:socialwelfare}a) being perturbed, which causes a larger deviation between the original problem and the perturbed problem.

Fig.~\ref{fig:comp2} further shows that the average approximation ratio is worse if a user submit more bids. 
The reason is clear: when the number of users $W$ is fixed, the more bids a user can submit, the larger the total number of bids $N$ is; hence due to the same reason as analyzed above, the ratio is closer to the theoretical one.

Fig.~\ref{fig:comp3} compares the average approximation ratio obtained in our experiments and the respective theoretical approximation ratios at different values of $\epsilon$, by plotting the relative approximation ratio $\frac{\text{average approx. ratio}}{1-\epsilon}$.
The average approximation ratio of our algorithm outperforms the theoretical ratio more when $\epsilon$ is larger, since $\theta^0_i,i=1,\ldots, N$ 
are larger (they are selected in the interval $[0,\epsilon / N]$) and the gap between the empirical approximation ratio and the theoretical one 
is larger according to (\ref{eqn:POPT_OPT_1}) in the proof of Lemma \ref{lem:POPT_vs_OPT}. 
The better performance at a smaller $N$ can be explained similarly as that for Fig.s \ref{fig:comp1}-\ref{fig:comp2}.

\begin{figure}[]
\begin{center}
  \includegraphics[width=0.48\textwidth]{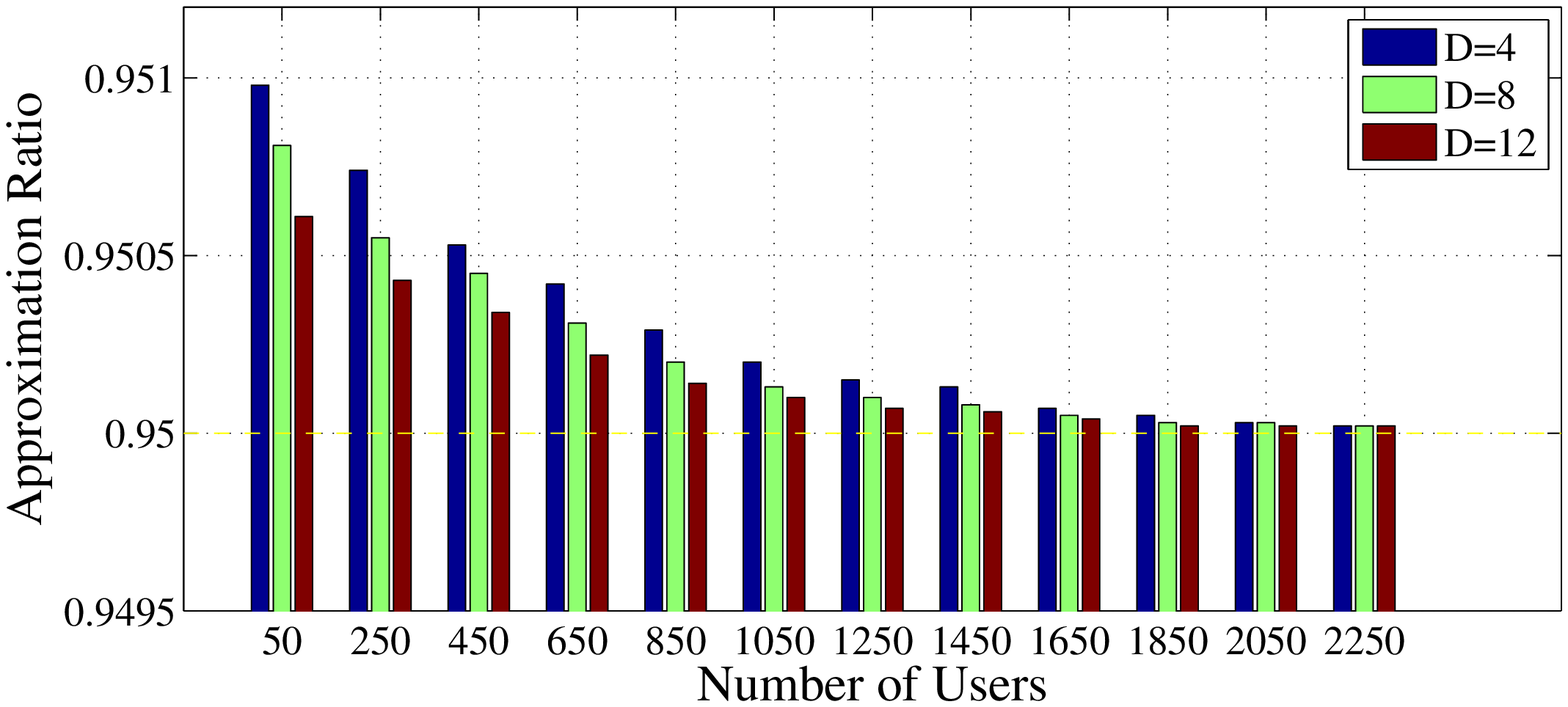}
\caption{Approx. Ratio of RPAA}
  \label{fig:comp1}
\end{center}
\vspace{-7mm}
\end{figure}
\begin{figure}[]
\begin{center}
  \includegraphics[width=0.48\textwidth]{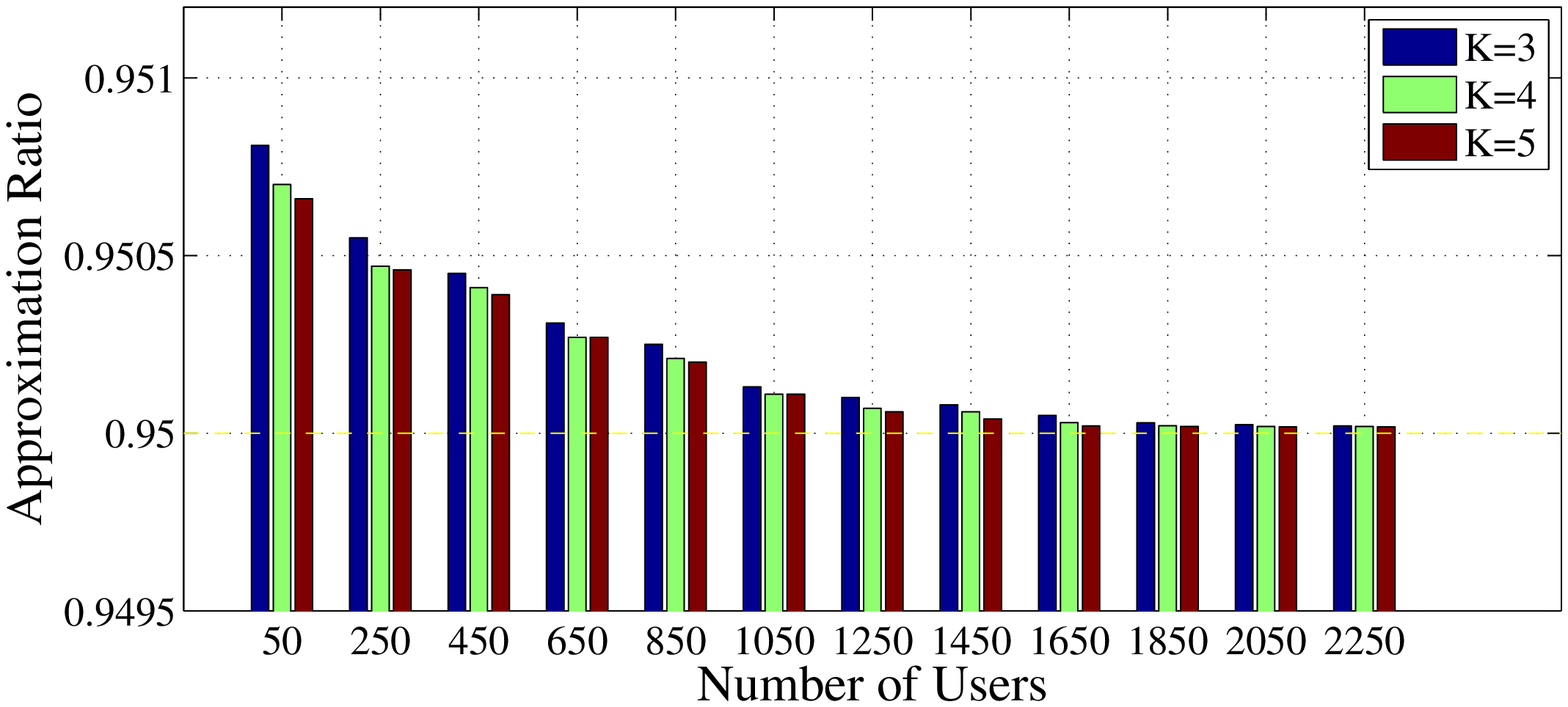}
\caption{Approx. Ratio of RPAA}
  \label{fig_k}
\end{center}
\vspace{-7mm}
\end{figure}

\begin{figure}[]
\begin{center}
  \includegraphics[width=0.48\textwidth]{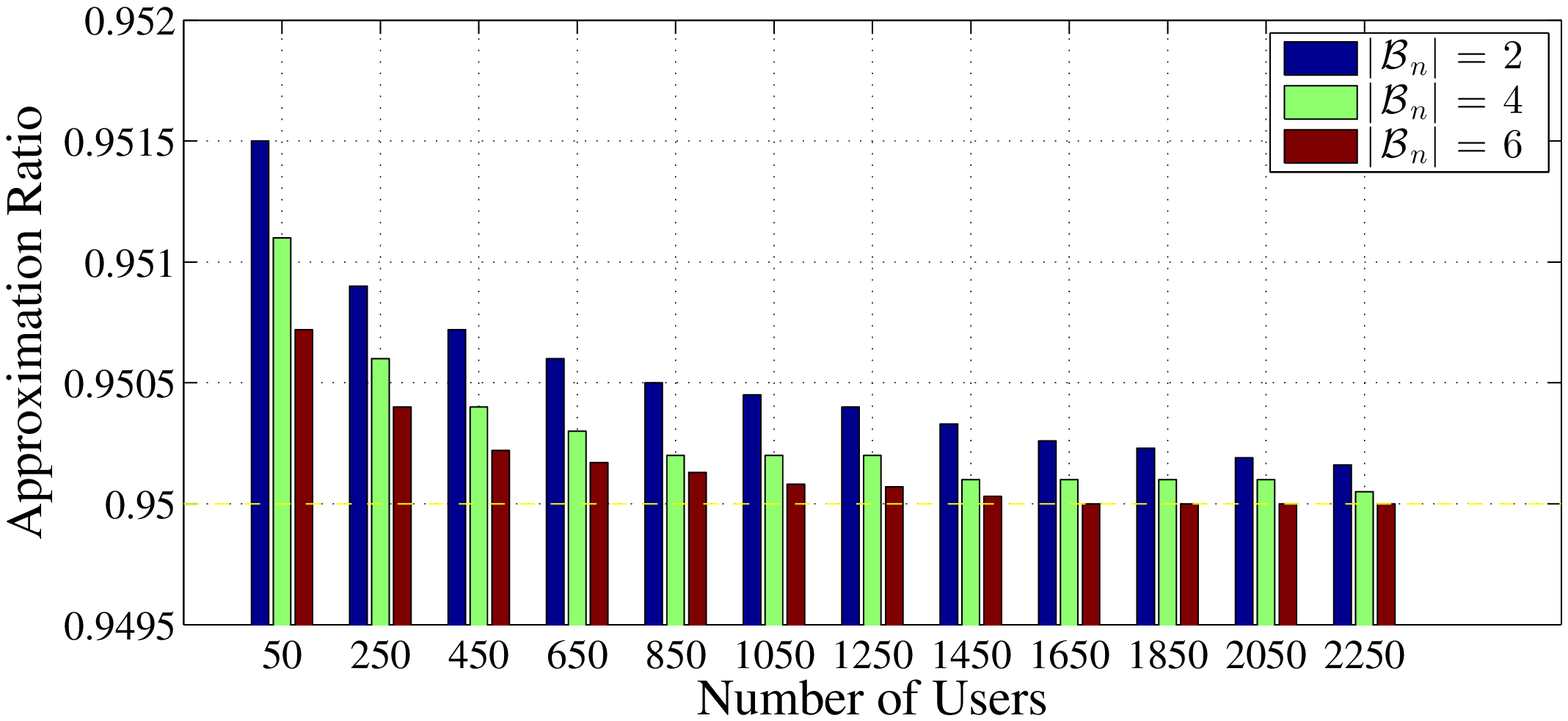}
  \caption{Approx. Ratio of RPAA}
  \label{fig:comp2}
\end{center}
\vspace{-7mm}
\end{figure}

\begin{figure}[]
\begin{center}
  \includegraphics[width=0.48\textwidth]{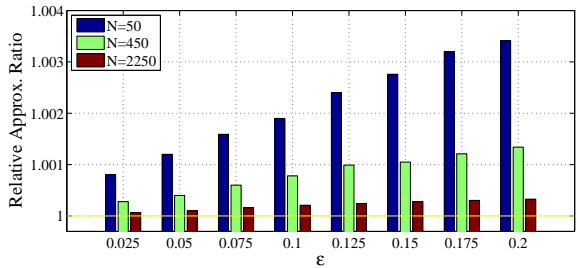}
  \caption{Relative Approx. Ratio of RPAA}
  \label{fig:comp3}
\end{center}
\vspace{-7mm}
\end{figure}

%

\subsection{Social Welfare Comparison with PDAA}

We now compare the social welfare achieved by our Alg.~\ref{alg:perturbation} with the primal-dual approximation algorithm in \cite{Zhang:2014jo} (which is essentially the algorithm used in \cite{Shi:sigmetrics14} as well), denoted by PDAA. The algorithm in \cite{Zhang:2014jo} does not consider the distribution of VM demands in multiple data centers. We hence extend this algorithm to multiple data centers by expanding the dimensions of the capacity constraint from $K$ to $K \times D$ to handle $K$ types of resources distributed in $D$ data centers, for a fair comparison. 

Fig.~\ref{fig:comp4} and Fig.~\ref{fig:comp5} show that our algorithm consistently outperforms the algorithm in \cite{Zhang:2014jo} in terms of social welfare, under the same parameter settings. This validates our theoretical analysis: our algorithm is guaranteed to achieve a no-lower-than-$(1-\epsilon)$ approximation in social welfare, where the other algorithm achieves a ratio of around $2.72$ \cite{Zhang:2014jo}.


Fig.~\ref{fig:comp4} also indicates that the social welfare of both algorithms increases with the increase of $W$ and $|\mathcal{B}_w|$. The resource capacity in our experiments is set to be roughly linear in the total resource demand of the users, and when $W$ is large, more bids can be accepted and hence the social welfare is larger. When each user can submit more bids, the decision space for ILP (\ref{eqn:socialwelfare}) is larger, leading to a better social welfare.

Fig.~\ref{fig:comp5} implies a negative correlation between the social welfare and $D$, which can be intuitively explained as follows: When the number of data centers is larger, the bid bundles that each user submits contain VMs scattered in more data centers. If any resource demand in any data center was not satisfied, a bid would be rejected. Moreover, as explained for Fig. \ref{fig:comp1} and Fig. \ref{fig_k}, more data centers lead to more constraints being perturbed (demands enlarged) in Alg.~2, which makes $\vec{x}^p$ deviate more from the optimum of the original problem. Note that there is a high probability for the final allocation solution $\vec{y}^{\epsilon}$ to be equal to $\vec{x}^p$. Based on all the above, the chance for each bundle to be accepted decreases, and the social welfare decreases slightly with the increase of $D$.

%
\captionsetup[figure]{labelfont=bf}
\begin{figure*}[t]
\captionsetup{width=0.48\textwidth}
\begin{center}
\begin{minipage}[t]{0.45\linewidth}
  \includegraphics[width=\textwidth]{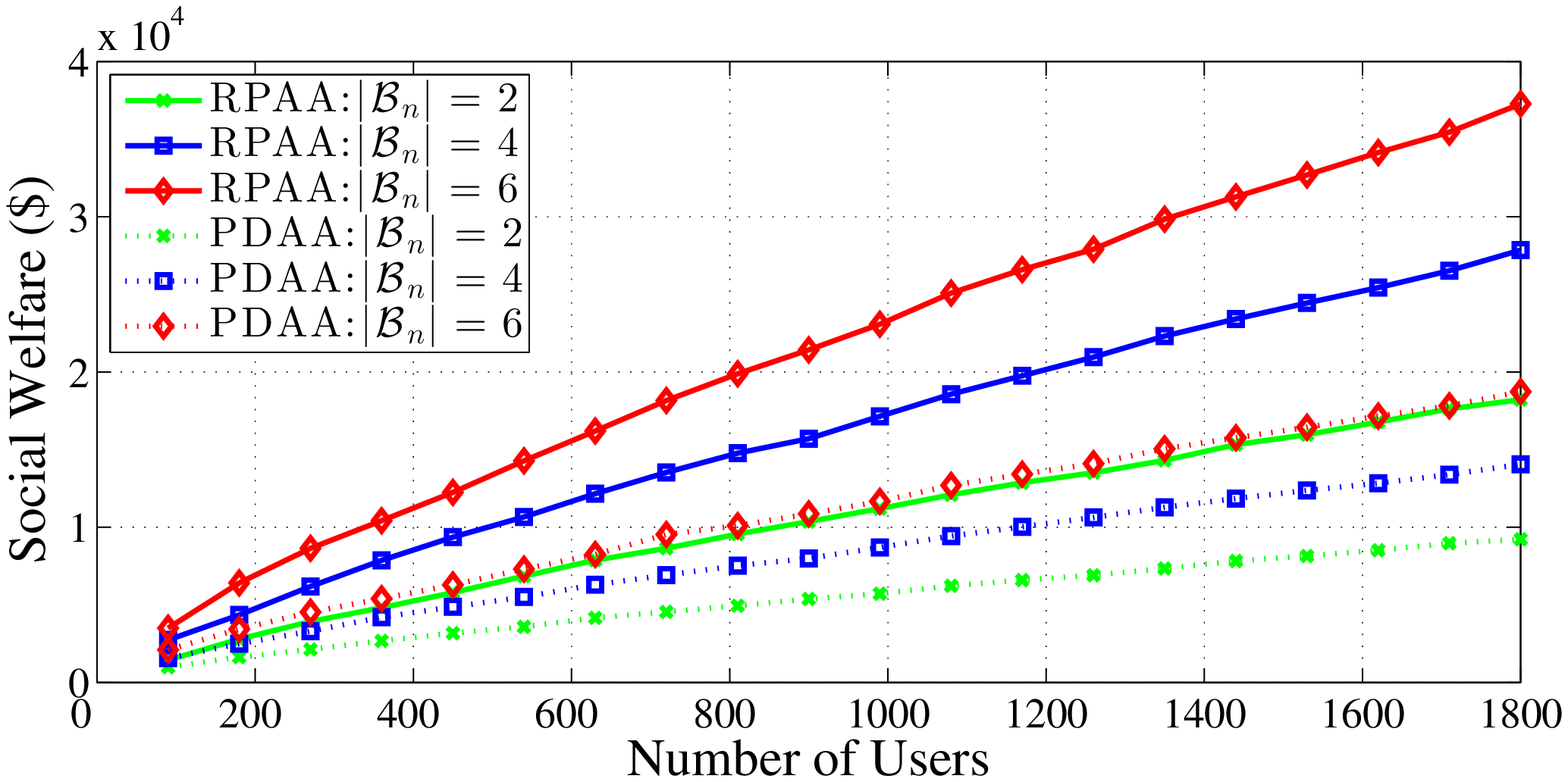}
  \caption{Social Welfare of RPAA and PDAA}
  \label{fig:comp4}
\end{minipage}%
\hfill
\begin{minipage}[t]{0.45\linewidth}
  \includegraphics[width=\textwidth]{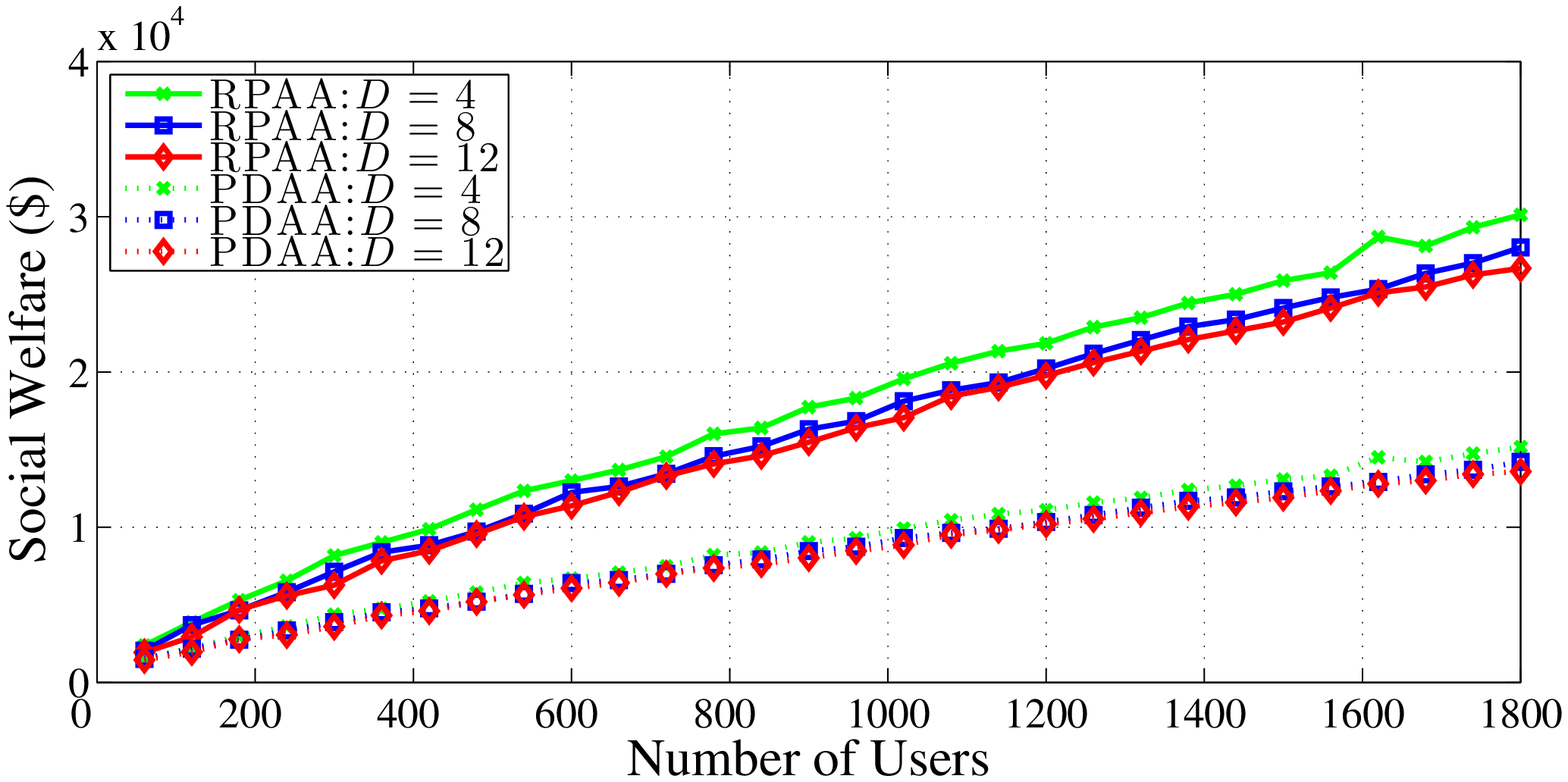}
  \caption{Social Welfare of RPAA and PDAA}
  \label{fig:comp5}
\end{minipage}
\end{center}
\vspace{-7mm}
\end{figure*}

\subsection{User Satisfaction Comparison with PDAA}

We next evaluate user satisfaction achieved by both RPAA and PDAA, which is the percentage of users accepted as winners in the respective auctions. Fig.~\ref{fig:comp6} and Fig.~\ref{fig:comp7} show that user satisfaction achieved by our algorithm is about twice that of the other algorithm, which results from similar reasons as given in the comparison of social welfare. User satisfaction of both algorithms improves slightly with the increase of the number of bids a user submits, mainly because more choices of the bids provide a user a higher chance to win one, while the chance does not improve much since all the users now have more bids to submit.
User satisfaction in Fig.~\ref{fig:comp7} decreases slightly as more data centers are included, suffering from the same cause as explained for Fig.~\ref{fig:comp5}.

\captionsetup[figure]{labelfont=bf}
\begin{figure*}[t]
\captionsetup{width=0.48\textwidth}
\begin{center}
\begin{minipage}[t]{0.45\linewidth}
  \includegraphics[width=\textwidth]{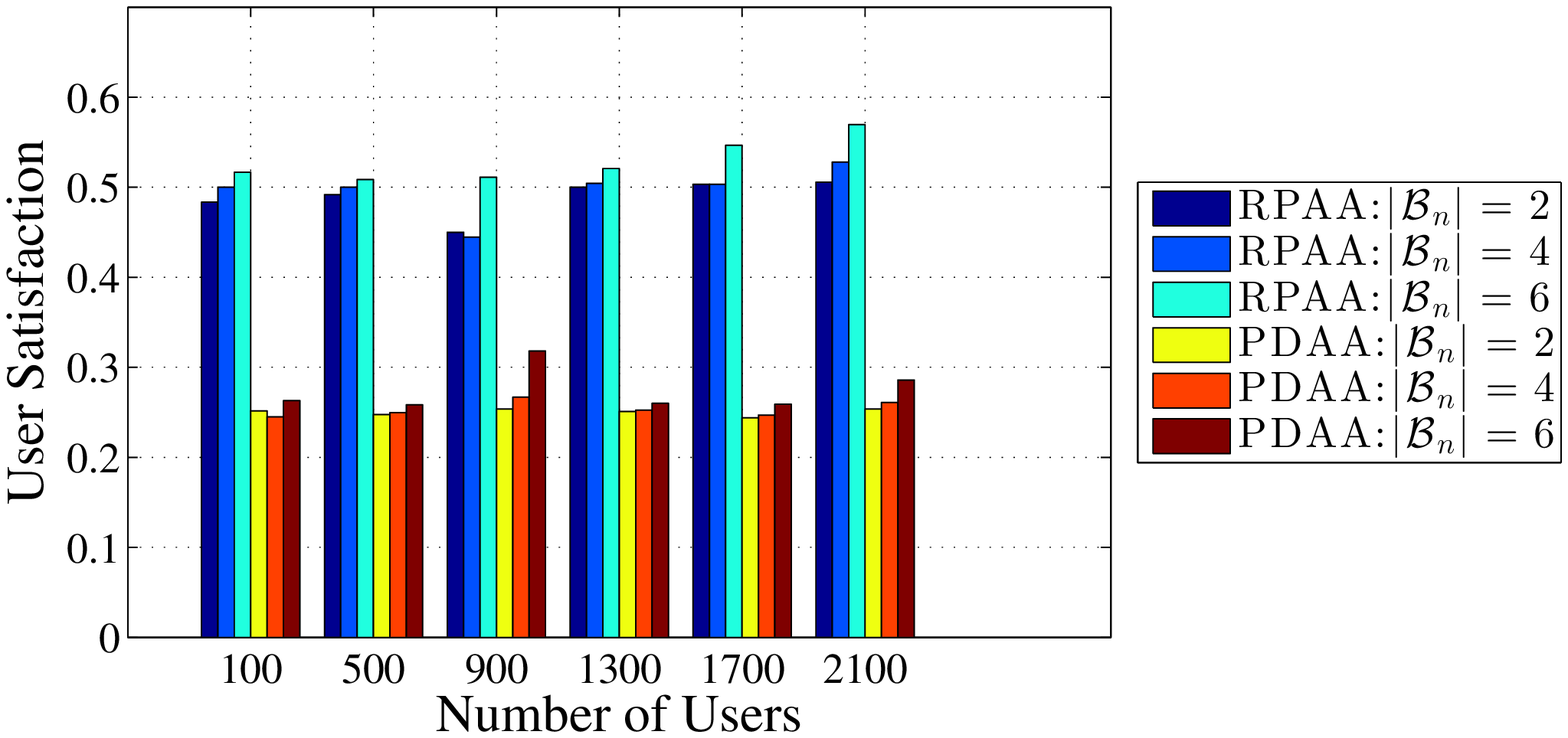}
  \caption{User Satisfaction of RPAA and PDAA}
  \label{fig:comp6}
\end{minipage}%
\hfill
\begin{minipage}[t]{0.45\linewidth}
  \includegraphics[width=\textwidth]{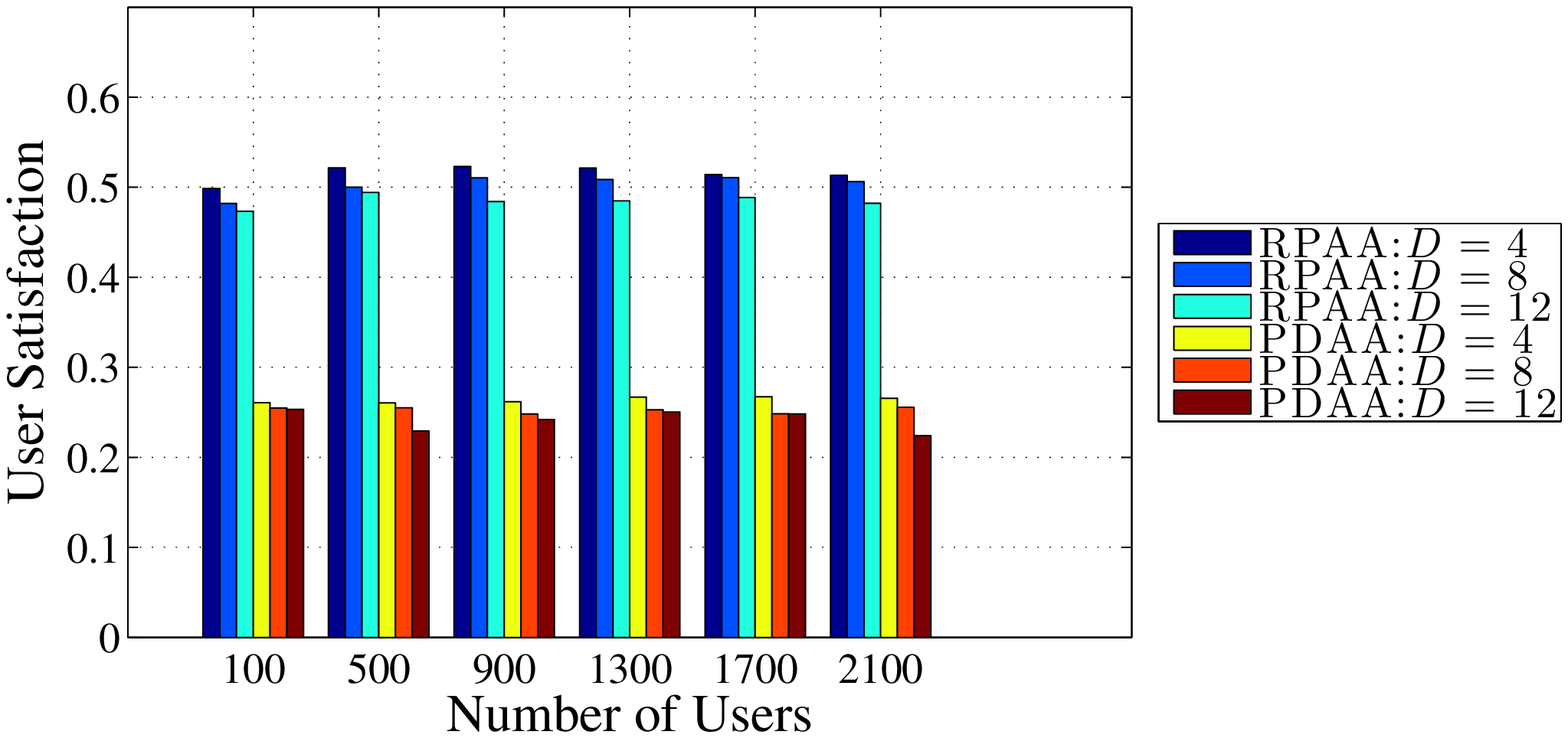}
  \caption{User Satisfaction of RPAA and PDAA}
  \label{fig:comp7}
\end{minipage}
\end{center}
\vspace{-7mm}
\end{figure*}

Finally, we remark on the running time incurred by the two algorithms. Figs. \ref{fig:comp4}--\ref{fig:comp7} illustrate that our algorithm outperforms the algorithm in \cite{Zhang:2014jo} in social welfare and user satisfaction, which is mainly due to the much better approximation ratio achieved by our algorithm. The running time of PDAA is
shown to be $O(N^6\log{}N)$, where $N$ is the number of bids, while our algorithm has an expected time complexity of $O(KDN^{8KD+1}/\epsilon^{2KD})$. (details in the proof of Theorem \ref{thm:expectedtime} in Appendix \ref{sec:thm:expectedtime}). Hence,
our algorithm may sacrifice some of the computation efficiency for a much better approximation to the optimal social welfare. However, the difference is not substantial, and a polynomial running time is still guaranteed for our algorithm in practice.

\section{Concluding Remarks}
\label{sec:conclusion}

This work presents a truthful and efficient auction mechanism for dynamic VM provisioning and pricing in geo-distributed cloud data centers. By employing smoothed analysis in a novel way and randomized reduction techniques, we develop a randomized mechanism that achieves truthfulness, polynomial running time, and $(1-\epsilon)$-optimal social welfare for resource allocation (all in expectation). 
We propose an exact algorithm which solves the NP-hard social welfare maximization problem in expected polynomial time, and apply a perturbation-based randomized scheme based on the exact algorithm to produce a VM provisioning solution that is $(1-\epsilon)$-optimal in social welfare in expectation. Combining the randomized scheme with a randomized VCG payment, we achieve an auction mechanism truthful in expectation. From a theoretical perspective, we achieve a randomized fully polynomial-time-in-expectation ($1-\epsilon$)-approximation scheme for a strongly NP-hard problem which does not have a deterministic FPTAS. We believe that this new technique can be generalized to work for a rich class of combinatorial auctions, other than VM auctions. Trace driven simulations we conduct validate our theoretical analysis and reveals the superior performance of our mechanism as compared to an existing mechanism on dynamic VM provisioning.

\ifCLASSOPTIONcaptionsoff
  \newpage
\fi

\bibliographystyle{IEEEtran}
{\small
\bibliography{IEEEabrv,main}
}
\opt{long}{\begin{appendices}

\section{Proof of Theorem \ref{thm:nphard}}
\label{sec:thm:nphard}
\begin{proof} 
Let 	

$R^i_{j}=\left\{
\begin{array}{l l}
		R_{kd}^i & \quad  j=1,\ldots,KD\\
		{1 \quad \mbox{\small if } i\in \mathcal{B}_{j-KD}}\atop
		{0 \quad \mbox{\small if } i\notin \mathcal{B}_{j-KD}}
		& \quad  j=KD+1, \ldots, KD+W
	\end{array}\right.$

\noindent and
	
$c_j=\left\{\begin{array}{l l}
c_{kd} & \quad  j=1,\ldots,KD\\
1& \quad j=KD+1,\ldots,KD+W
  \end{array}\right.$

\noindent The social welfare maximization problem (\ref{eqn:socialwelfare}) can be converted to a Multi-dimensional Knapsack Problem in polynomial time:

\begin{equation}
	{\rm maximize} \hspace*{2mm} \sum_{i=1}^{N}{b}_i x_i \label{eqn:knapsack}
\end{equation}

subject to:

\begin{align*}
\sum_{i=1}^{N}x_iR^i_{j}\leq c_j,&& j=1,\ldots, KD+W,\\
x_i\in \{0,1\},&&\forall i\in [N]
\end{align*}
	

The Multi-dimensional Knapsack Problem is a classic strongly NP-hard combinatorial optimization problem \cite{Puchinger:2010:MKP:1791474.1791481}. Moreover, it has no FPTAS unless P=NP \cite{Gens:1980:CAA:1008861.1008867}\cite{Kulik:2010:NET:1824821.1824869}. 
\end{proof}

\section{Proof of Lemma \ref{lem:recursion}}
\label{sec:lem:recursion}

\begin{proof}
If $\vec{x}^{(i)} \in \mathcal{P}(i)$, then $\vec{x}^{(i)}$ is not dominated by any solutions in $\mathcal{P}(i)$. If $\vec{x}^{(i)}-x^{(i)}_i$ ({\em i.e.}, the vector obtained by removing the last element $x^{(i)}_i$ from $\vec{x}^{(i)}$) is dominated by some solution in $\mathcal{P}(i-1)$, {\em e.g.}, $\vec{x}_d^{(i-1)}$, then $\vec{x}_d^{(i-1)}+x^{(i)}_i$ ({\em i.e.}, the vector obtained by appending $x^{(i)}_i$ to the end of $\vec{x}_d^{(i-1)}$) dominates $\vec{x}^{(i)}$, which is a contradiction. Thus we can conclude that if $\vec{x}^{(i)}\in \mathcal{P}(i)$, then $\vec{x}^{(i)}-x^{(i)}_i\in \mathcal{P}(i-1)$.\\
\end{proof}

\section{Proof of Lemma \ref{nondecreasingparetoset}}
\label{sec:lem:nondecreparetoset}
 
\begin{proof}
According to the construction from $\mathcal{P}(i)^{\prime}$ to $\mathcal{P}(i)$ in Algorithm 1, some solutions of $\mathcal{P}(i)^{\prime}$ will be pruned. 
 When we merge $\mathcal{P}(i-1)+0$ and $\mathcal{P}(i-1)+1$ to $\mathcal{P}(i)^{\prime}$, (1) if all the solutions in $\mathcal{P}(i-1)+0$ are kept in $\mathcal{P}(i)^{\prime}$, then $|\mathcal{P}(i)^{\prime}|\geq |\mathcal{P}(i-1)+0|$; (2) Otherwise, if $\exists \vec{x}_a^{(i)}\in\mathcal{P}(i-1)+0$ which is pruned, then $\exists \vec{x}_b^{(i)}\in \mathcal{P}(i-1)+1$ that dominates $\vec{x}_a^{(i)}$, which means there always exists a solution which substitutes the removed one. Hence $|\mathcal{P}(i)|\geq |\mathcal{P}(i-1)|.$
\end{proof}

\section{Proof of Lemma \ref{lem:POPT_vs_OPT}}
\label{sec:lem:POPT_vs_OPT}
\begin{proof}
We define solution $\vec{x}^{\ast -}$ to be a feasible solution to the perturbed problem which is no larger than $\vec{x}^{\ast}$ component-wisely. 
Due to $\vec{\hat{b}}=P \vec{b}$, we have 
\begin{equation}
POPT=(P\vec{b})^T\vec{x}^p \ge (P\vec{b})^T\vec{x}^{\ast -},
\end{equation}
\noindent since $\vec{x}^{\ast -}$ is a feasible solution to the perturbed problem. In the following, we will show that the perturbed objective value under $\vec{x}^{\ast -}$ is at least a $(1-\Delta)$-fraction of the perturbed objective value under $\vec{x}^{\ast}$, which is further at least a $(1-\Delta)(1-\epsilon/2)$-fraction of $OPT$, the original objective value: 
\begin{align}
&(P\vec{b})^T\vec{x}^{\ast -}\ge (1-\Delta)(P\vec{b})^T\vec{x}^{\ast}\notag\\
\ge& (1-\Delta)\sum_{i\in [N]}((1-\epsilon/2)b_i+\frac{\epsilon\sum_{i^{\prime}\in[N]}b_{i}}{N})x^{\ast}_i\label{eqn:POPT_OPT_1}\\
\ge& (1-\Delta)(1-\epsilon/2)\vec{b}^T\vec{x}^{\ast} \label{eqn:POPT_OPT}
\end{align}
The rest of the proof is to show that we could upper-bound $\Delta \le \epsilon/2$ by assuming that each bid only requests for a small amount of demand for each type of resource. Let $\mathcal{S}^{\ast}$ denote the set of the accepted bids of $\vec{x}^{\ast}$. To find a $\vec{x}^{\ast-}$, a feasible solution to the perturbed problem, our method is to construct a feasible solution to the perturbed problem by dropping some bids from the accepted bids of $\vec{x}^{\ast}$, until all the constraints are feasible again. This idea is directed by the fact that the demand parameters are only slightly perturbed, thus dropping a small amount of demand from the original optimal solution may only lose little social welfare.     
Note that such $\vec{x}^{\ast-}$ may not be unique, but we only need to find a possible one and lower bound the derived perturbed objective value under it. 

\begin{table}[h]
\caption{Dropping-based-on-Sorting Algorithm $DoS$}
\label{table:alg_dos}
\centering
\begin{tabular}{|l|}
\hline
1: Identify $\mathcal{Q}^{+}$, the set of the violated constraints;\\
2: For each $j\in \mathcal{Q}^{+}$:\\
3:~~~~Sort all the bids of the original optimal set $\mathcal{S}^{\ast}$\\
~~~~~~in the increasing order of $\frac{\hat{b}_i}{\hat{R}^j_{i}}$;\\
4:~~~~Drop the sorted bids one by one, \\
~~~~~~util the constraint $j$ is feasible.\\
\hline
\end{tabular}
\end{table}

To prove that the perturbed objective value under $\vec{x}^{\ast-}$ is at least $1-\epsilon/2$ of that under $\vec{x}^{\ast}$ is equivalent to show that the total bidding price of dropped bids is at most a $\epsilon/2$ fraction of the total bidding price aggregated by the accepted bids of $\vec{x}^{\ast}$. Before that, we define $\mathcal{Q}$ to be the set of $j$ (recall all the tuples of $(k,d)$ are sorted and re-indexed by $j$). Let $\mathcal{Q}^{+}$ be the subset of $\mathcal{Q}$ where at each element $j \in \mathcal{Q}^{+}$, the total perturbed demand of dimension $j$ exceeds the respective capacity. Recall that in the algorithm $DoS$, we drop bids in $|\mathcal{Q}^+|$ rounds to make the violated constraints feasible. A critical concept $\mathcal{S}^{-}_{j}$ is defined to be the set of dropped bids in the round of $j$. Three facts are shown in \eqref{eqn:fact1} and \eqref{eqn:fact2}.
\begin{align}
&\sum_{i\in \mathcal{S}^{-}_{j}}\hat{R}^{j}_i \le (1+2\epsilon) R_{max}^{j}, &\forall j \in \mathcal{Q}^{+} \label{eqn:fact1}\\
&\sum_{i\in \mathcal{S}^{\ast}} \hat{R}_{i}^{j} > c_{j}, \forall j \in \mathcal{Q}^{+}; ~\sum_{i\in \mathcal{S}^{\ast}} \hat{R}_{i}^{j} \le c_{j}, &\forall j \in \mathcal{Q}\setminus \mathcal{Q}^{+}\label{eqn:fact2}
\end{align} 
\noindent Here, fact in \eqref{eqn:fact2} follows the definitions of $\mathcal{Q}^{+}$ and $\mathcal{Q}$. 
The fact in \eqref{eqn:fact1} is derived as follows. 
Recall that the perturbation of demand parameter is $\hat{R}^{j}_{i}=R^{j}_i+\frac{\theta^j_{i}\sum_{i^{\prime}\in[N]}R^{j}_{i}}{N}$ where $\theta^j_i \sim U(0, \frac{\epsilon}{N})$. 
We define the maximal demand of $k$-type resource in data center $d$ to be $R^{j}_{max}=\max_{i\in [N]}R^{j}_i$. 
Then we have 
\begin{equation}
\sum_{i\in \mathcal{S}^{\ast}} \hat{R}^{j}_i\le \sum_{i\in \mathcal{S}^{\ast}} R^{j}_{i} +\sum_{i\in \mathcal{S}^{\ast}} \frac{\epsilon}{N}R^{j}_{max} \le \sum_{i\in \mathcal{S}^{\ast}}R^{j}_{i}+\epsilon R^{j}_{max}.
\end{equation} 
Moreover, we have the total demand under the optimal solution of the original problem is at most the capacity of each type of resource at each data center, {\em i.e.,} $\sum_{i\in \mathcal{S}^{\ast}}R^{j}_i\le c_{j}$ for each $j \in \mathcal{Q}$. 
Thus, for each $j\in \mathcal{Q}^{+}$, we can upper-bound the total perturbed demand under $\vec{x}^{\ast}$, to be $\sum_{i\in \mathcal{S}^{\ast}} \hat{R}^{j}_i \le c_{j}+\epsilon R^{j}_{max}$. It means that if the dropped amount is at least $\epsilon R^{j}_{max}$ of each $j\in \mathcal{Q}^{+}$, the packing constraints will be feasible again.
Now the question is, in each round, how much demand will be actually dropped at most? Consider at some time in round $j$ when the dropping is still not finished yet. We must have that the dropped $j$th demand is less than $\epsilon R^{j}_{max}$; otherwise the remaining $j$th demand does not exceed $c_{j}$, the capacity of $j$th demand, which means this round of dropping should have been stopped. Then consider the moment in this round that the total $j$th demand that have been dropped is less than $\epsilon R^{j}_{max}$, but the bid to be dropped next will finish this round of dropping. We call this moment {\em critical moment}. Note that this critical moment always exists in each dropping round. Since the very next dropped bid in the critical moment is at most $\hat{R}^{j}_{max}\le (1+\epsilon)R^j_{max}$. Plus $\epsilon R^{j}_{max}$, the dropped amount before the critical moment, we prove the fact in \eqref{eqn:fact1}.

The above three facts essentially mean that we find a $\vec{x}^{\ast -}$, a feasible solution of the perturbed problem, such that the allocated resource amount is not much less than that under $\vec{x}^{\ast}$, the optimal solution to the original problem. Recall that to lower-bound $\Delta$, the remaining is to prove that the total perturbed bidding price generated by remaining bids is not much less than that of $\vec{x}^{\ast}$. How to associate the remaining demand to the the remaining social welfare? In fact, the sorting operation gives the answer. Recall in each round $j\in \mathcal{Q}^{+}$, we sort and drop the bids of $\mathcal{S}^{\ast}$ in the increasing order of $\frac{\hat{b}_i}{\hat{R}^{j}_i}$, the bidding price per unit of $j$th demand. It means that the dropped bids have lower bidding price per unit of demand $j$ than the average $\frac{\hat{b}_i}{\hat{R}^{j}_i}$ over all the bids in $\mathcal{S}^{\ast}$. The basic idea to upper-bound the dropped social welfare is as follows. In each round $j$, the dropped $j$th demand could be upper-bounded according to the fact in \eqref{eqn:fact1}. Then we could upper-bound the fraction of the loss of perturbed social welfare due to the dropped bids in the $j$th round over the total perturbed social welfare. Finally, taking over all the rounds, we upper-bound the fraction of the total loss of perturbed social welfare over that under $\vec{x}^{\ast}$, which is shown as follows.
\begin{align}
&\sum_{j\in \mathcal{Q}^{+}} \frac{\sum_{i\in \mathcal{S}^{-}_{j}} \hat{b}_i}{\sum_{i\in \mathcal{S}^{\ast}}\hat{b}_i}\le \sum_{j\in \mathcal{Q}^{+}} \frac{\sum_{i\in \mathcal{S}^{-}_{j}} \hat{R}^j_i}{\sum_{i\in \mathcal{S}^{\ast}}\hat{R}^j_i}~(\text{due to sorting})\\
&\le \sum_{j\in \mathcal{Q}^{+}} \frac{(1+2\epsilon) R_{max}^{j}}{c_{j}}~(\text{due to:}~\eqref{eqn:fact1}~\text{and}~\eqref{eqn:fact2})\label{eqn:delta_bound}\\
&\le KD(1+2\epsilon)\times \frac{1}{2KD(2+\frac{1}{\epsilon})}=\frac{\epsilon}{2}\label{eqn:delta_bound_2}
\end{align}
\noindent The first inequality in \eqref{eqn:delta_bound_2} is due to the small bid assumption stated in Lemma \ref{lem:POPT_vs_OPT}, $\frac{R^j_{max}}{c_j}\le \frac{1}{2KD(2+\frac{1}{\epsilon})}$. The assumption essentially means that each bid has a small demand of each type of resource in the desired data center, compared to the corresponding capacity. Thus, $1-\Delta$ in \eqref{eqn:POPT_OPT}, the fraction of remaining perturbed social welfare over that under $\vec{x}^{\ast}$ is lower-bounded to be $1-\frac{\epsilon}{2}$.

According to \eqref{eqn:POPT_OPT}, we have that the social welfare of the perturbed problem under $\vec{x}^{-}$ over that under $\vec{x}^{-}$ is at least $(1-\Delta)(1-\epsilon/2)\ge (1-\epsilon/2)(1-\epsilon/2)\ge 1-\epsilon$, which serves as a critical step to lower-bound the approximation ratio of our allocation algorithm.
 Note that our analysis method is valid for milder conditions that require a higher upper-bound of the fraction of bid demand over capacity. Nevertheless, to obtain the $(1-\epsilon)$-approximation ratio, we will consider such a small bid assumption which is stated in Lemma \ref{lem:POPT_vs_OPT}, Theorem \ref{thm:smoothedfptas} and Theorem \ref{thm:truthful}.

\end{proof}

\section{Proof of Theorem \ref{thm:boundedtime}}
\label{sec:thm:boundedtime}
\begin{proof}

There are two main steps in each $i$th round:\\
1. Given $\mathcal{P}(i-1)$, we construct $\mathcal{P}(i)^{\prime}$ by using two copies of all the solutions $\vec{x}^{(i-1)}\in \mathcal{P}(i-1)$ and adding the $i$th bid to each solution of one copy with the value of $0$ and $1$, respectively. The time this step takes is linear in $|\mathcal{P}(i-1)|$, increasing the total bidding price and resource demands of all the solutions in $\mathcal{P}(i-1)$. 

2. For each solution $\vec{x}^{(i)} \in \mathcal{P}(i)^{\prime}$, we compare $(b(\vec{x}^{(i)}),C_{kd}(\vec{x}^{(i)}),\forall k \in [K], \forall d \in [D])$ to all the other solutions in $\mathcal{P}(i)^{\prime}$ and delete all the dominated solutions. Thus the time this step takes is linear in $KD|\mathcal{P}(i)^{\prime}|^2$.\\

According to the analysis in step 1 and step 2, the running time is bounded by:\\ 

$~~~~~~~~~~~~O(KD\sum_{i=1}^{N-1}|\mathcal{P}(i)|^{2}).$\\

Due to $Lemma~2$, $O(KD\sum_{i=1}^{N-1}|\mathcal{P}(i)|^2)\leq O(KDN|\mathcal{P}(N)|^2).$\\ 
\end{proof}

\section{Proof of Theorem \ref{thm:expectednumber}}
\label{sec:thm:expectednumber}
\begin{proof}
This theorem desires to upper-bound the number of Pareto optimal solutions in expectation when parameters in \eqref{eqn:socialwelfare_p} are independently and randomly perturbed. 
We adopt classical smoothed analysis in which the expected number of Pareto optimal solutions usually relies on the input size and the maximal perturbation density of the perturbed input parameters. Let $\phi$ denote the upper-bound of the perturbation density of each $\hat{b}_i$ and $\hat{R}^{kd}_i$. 
Now we calculate $\phi$ of our perturbed problem. 
Since $\theta^j_i$ is drawn from $[0,\frac{\epsilon}{N}]$, the value of $\hat{b}_i$ lies in the interval $[(1-\epsilon){b_i},(1-\epsilon){b_i}+\frac{\epsilon \sum_{j=1}^N{{b}_j}}{N^2}]$, with the interval length of $\frac{\epsilon \sum_{i^{\prime}=1}^N{{b}_{i^{\prime}}}}{N^2}$. Let $b_{max}= \max\{b_1, \ldots, b_N\}$.  
We have $\sum_{i^{\prime}=1}^N{{b}_{i^{\prime}}}\geq b_{max}$. Thus the interval length is no smaller than $\frac{\epsilon b_{max}}{N^2}$ and equivalently, the density of $\hat{b}_i$ is upper-bounded everywhere in the interval $[(1-\epsilon){b_i},(1-\epsilon){b_i}+\frac{\epsilon \sum_{i^{\prime}=1}^N b_{i^{\prime}}}{N^2}]$ by $\frac{N^2}{\epsilon b_{max}}$. Similarly, the interval length of the perturbation of $\hat{R}^{j}_i$ is no smaller than $\frac{\epsilon R^j_{max}}{N^2}$, and the density of $\hat{R}^j_i$ is upper-bounded everywhere in the interval $[R^j_i, R^j_i+ \frac{\sum_{i^{\prime}=1}^N R^j_{i^{\prime}}}{N^2}]$ by $\frac{N^2}{\epsilon R^j_{max}}$. Moreover, without loss of generality, we can normalize all the bidding prices to be in $[0,1]$ by dividing each $b_i$ by $b_{max}$. We also normalize all the demands to be in $[0,1]$ by diving each $R^j_i$ by $R^j_{max}$ and diving each $c_j$ by $R^j_{max}$. The perturbation density will be upper-bounded by $\frac{N^2}{\epsilon }$.

According to the latest result of the smoothed number of Pareto optimal solutions for a multi-objective integer programming \cite{Brunsch:improved_smoothed}, we obtain the result in the Proposition \ref{prop:cite}. Before showing the proposition, we show that the result of a multi-objective optimization problem works for our problem. 

According to the definition of a Pareto optimal solution to a multi-objective integer programming  \cite{Brunsch:improved_smoothed}, a solution $\vec{x}$ in the feasible region is Pareto optimal if and only if there is no solution which is at least as good as $\vec{x}$ and better than $\vec{x}$ in at least one criteria. That is, a Pareto optimal solution is a feasible solution that is not dominated by any other solution in the feasible region. In a multi-objective optimization problem, there are multiple objective functions to be optimized in a feasible region. In our problem, according to the definition of Pareto optimal solution (Pareto Optimal Allocation), the defined objectives (criteria to compare two solutions, not the objective function) are $C_{kd}(\vec{x})'s$, the total demand for each $k$-type of resource in data center $d$, plus $s(\vec{x})$, the social welfare under $\vec{x}$. It means we have in total $KD+1$ objects of criteria to define a Pareto optimal solution. In our problem, a solution $\vec{x}$ is a Pareto optimal solution if and only if there does not exist a feasible solution $\vec{x}^{\prime}$ that dominates $\vec{x}$ in all objects of criteria, which falls in the family of the Pareto optimal solution in \cite{Brunsch:improved_smoothed}. Besides the fact that the definition of a Pareto optimal solution for a multi-objective optimization problem is essentially the same as ours, the perturbation policy of our work fits that of \cite{Brunsch:improved_smoothed} where: Each coefficient ($\hat{b}_i$ and $\hat{R}^j_i,~\forall j\in [KD-1]$) is chosen independently according to its own quasi-concave probability density function where the density at each point of the perturbation interval is upper-bounded by $\phi$. Here, quasi-concave density \cite{Brunsch:improved_smoothed} requires that the probability density is non-decreasing within the left half perturbation interval while is non-increasing in the right half. Our perturbation of each parameter follows the uniform distribution which has a quasi-cave density function.   
\begin{proposition}\label{prop:cite}
For any constant $K$, $D$ and $\alpha \in \mathbb{N}$, the $\alpha$th moment of the smoothed number of Pareto optimal solutions of a multi-objective binary programming is $O((N^{2KD}\phi^{KD})^{\alpha})$ for quasi-concave perturbation density functions with density everywhere upper-bounded by $\phi$. 
\end{proposition}
\noindent Here, according to definition of $\alpha$th moment of a random variable \cite{wiki:moment}, $\alpha=2$ is the case to calculate the expectation of $|\mathcal{P}(N)|^2$.  Thus putting $\phi = \frac{N^2}{\epsilon}$ and $\alpha =2$ into Proposition \ref{prop:cite}, we have that $E[|\mathcal{P}(N)|^2]\le O(N^{8KD}/\epsilon^{2KD})$, which the theorem follows.   

\end{proof}

\section{Proof of Theorem \ref{thm:expectedtime}}
\label{sec:thm:expectedtime}

\begin{proof}

Combining Theorem \ref{thm:boundedtime} and Theorem \ref{thm:expectednumber}, we derive the expected running time of our exact algorithm on the randomly perturbed ILP (\ref{eqn:socialwelfare_p}) (line 5 in Alg.~\ref{alg:perturbation}) as

\begin{align*}
	O(KDNE[|\mathcal{P}(N)|^2])=O(KDN^{8KD+1}/\epsilon^{2KD}).
\end{align*}

\noindent Note that $K$ and $D$ are fixed constants. The perturbation matrix $P$ can be obtained in polynomial time (lines 3--4 in Alg.~\ref{alg:perturbation}). The set of feasible solutions to ILP (\ref{eqn:socialwelfare}) and the distribution $\Omega$ can be constructed in polynomial time as well (lines 6 in Alg.~\ref{alg:perturbation}). Hence the theorem is proven.
\end{proof}

\end{appendices}}

\end{document}